\documentclass[a4paper,UKenglish]{lipics-v2019}
\usepackage{microtype}%

\usepackage[utf8]{inputenc}

\usepackage{xspace}
\usepackage{color}
\usepackage[usenames,dvipsnames,svgnames,table]{xcolor}
\usepackage{tikz}
\usetikzlibrary{decorations.markings}

\usepackage{mathrsfs}
\usepackage{amsmath, amsthm, amssymb,latexsym}
\usepackage{thm-restate}

\usepackage{url}
\usepackage{multirow}

\usepackage{float}
\newfloat{Algorithm}{t}{lop}

\usepackage[noend]{algpseudocode}   %

\theoremstyle{plain}

\newtheorem{observation}{Observation}

\newcommand\Z{\ensuremath{\mathbb{Z}}} %

\newcommand{\tin}[1]{^{(#1)}}
\newcommand{\sm}{{\mbox{\footnotesize small}}}
\newcommand{\I}{\mathcal{I}}
\newcommand{\A}{\mathcal{A}}

\newcommand{\bF}{\mathbf{f}}
\renewcommand{\S}{\mathcal{S}}

\newcommand{\network}{\mathcal{N}}
\newcommand{\networknt}{\mathcal{N}_{n,t}}
\newcommand{\frag}{fragile complexity\xspace}
\newcommand{\frags}{fragile complexities\xspace}

\newcommand{\p}{\mathbf{Pr}}
\newcommand{\pr}[1]{\ensuremath{\mathbf{Pr}\left[#1\right]}}

\newcommand{\work}{work\xspace}
\newcommand{\sorting}{\ensuremath{(n)}-sorting\xspace}
\newcommand{\select}{{\sc Selection}\xspace}
\newcommand{\selection}[1]{\ensuremath{(n,#1)}-selection\xspace}
\newcommand{\partition}[1]{\ensuremath{(n,#1)}-partition\xspace}
\newcommand{\rank}[1]{\ensuremath{rank(#1)}}
\newcommand{\rankx}[1]{\ensuremath{rank_X(#1)}}
\newcommand{\f}[1]{\ensuremath{f({#1})}}
\newcommand{\w}[1]{\ensuremath{w({#1})}}

\newcommand{\sig}{\ensuremath{\sigma}}
\newcommand{\E}[1]{\ensuremath{\mathbb{E}\left[#1\right]}}
\newcommand{\Oh}{\ensuremath{\mathcal{O}}}

\newcommand{\markR}{\ensuremath{\mathcal{R}}\xspace}
\newcommand{\markL}{\ensuremath{\mathcal{L}}\xspace}

\newcommand{\fmed}{\ensuremath{f_\text{med}}}
\newcommand{\fmin}{\ensuremath{f_\text{min}}}
\newcommand{\frem}{\ensuremath{f_\text{rem}}}

\newcommand{\randmedname}{RMedian}
\newcommand{\randmed}{\textsc{\randmedname}}

\newcommand{\ceil}[1]{{\left\lceil{#1}\right\rceil}}
\newcommand{\floor}[1]{{\left\lfloor{#1}\right\rfloor}}

\let\emptyset\varnothing

\DeclareMathOperator{\polylog}{polylog}

\newif\ifComments
\Commentstrue
\ifComments
\newcommand{\peyman}[1]{{\color{purple}\#\#{(P:)\footnotesize{ #1 }}\#\#}}
\newcommand{\irina}[1]{{\color{purple}\#\#{(I:)\footnotesize{ #1 }}\#\#}}
\newcommand{\manuel}[1]{{\color{purple}\#\#{(M:)\footnotesize{ #1 }}\#\#}}
\newcommand{\david}[1]{{\color{purple}\#\#{(D:)\footnotesize{ #1 }}\#\#}}
\newcommand{\riko}[1]{{\color{purple}\#\#{(R:)\footnotesize{ #1 }}\#\#}}
\newcommand{\nodari}[1]{{\color{purple}\#\#{(N:)\footnotesize{ #1 }}\#\#}}
\newcommand{\uli}[1]{{\color{purple}\#\#{(U:)\footnotesize{ #1 }}\#\#}}
\else
\newcommand{\peyman}[1]{}
\newcommand{\irina}[1]{}
\newcommand{\manuel}[1]{}
\newcommand{\david}[1]{}
\newcommand{\riko}[1]{}
\newcommand{\nodari}[1]{}
\newcommand{\uli}[1]{}
\fi

\newcommand{\ignore}[1]{}

\title{\Large Fragile Complexity of Comparison-Based Algorithms\footnote{%
This material is based upon work performed while attending AlgoPARC Workshop on Parallel Algorithms and Data Structures at the University of Hawaii at Manoa, in part supported by the National Science Foundation under Grant No. 1745331.\\
This work was also partially supported 
by the Deutsche Forschungsgemeinschaft (DFG) under grants ME~2088/3-2 and ME~2088/4-2,
by the Independent Research Fund Denmark, Natural Science, under grant DFF-7014-00041,
and by the National Science Foundation under Grant No. CCF-1533823.
}}
\titlerunning{\normalsize Fragile Complexity of Comparison-Based Algorithms}

\author{Peyman Afshani}{Aarhus University}{peyman@cs.au.dk}{}{}
\author{Rolf Fagerberg}{University of Southern Denmark}{rolf@imada.sdu.dk}{}{}
\author{David Hammer}{Goethe University Frankfurt and University of Southern Denmark}{hammer@imada.sdu.dk}{}{}
\author{Riko Jacob}{IT University of Copenhagen}{rikj@itu.dk}{}{}
\author{Irina Kostitsyna}{TU Eindhoven}{i.kostitsyna@tue.nl}{}{}
\author{Ulrich Meyer}{Goethe University Frankfurt}{umeyer@ae.cs.uni-frankfurt.de}{}{}
\author{Manuel Penschuck}{Goethe University Frankfurt}{mpenschuck@ae.cs.uni-frankfurt.de}{}{}
\author{Nodari Sitchinava}{University of Hawaii at Manoa}{nodari@hawaii.edu}{}{}
\authorrunning{\normalsize P. Afshani, F. Fagerberg, D. Hammer, R. Jacob, U. Meyer, M. Penschuck, N. Sitchinava}

\Copyright{Peyamn Afshani, Rolf Fagerberg, David Hammer, Riko Jacob, Irina Kostitsyna, Ulrich Meyer, Manuel Penschuck, and Nodari Sitchinava}

\acknowledgements{We thank Steven Skiena for posing the original problem, and we thank Michael Bender and Pat Morin for helpful discussions.}

\nolinenumbers
\hideLIPIcs

\begin{document}

\maketitle

\begin{abstract}
We initiate a study of algorithms with a focus on the computational
complexity of individual elements, and introduce the \emph{\frag} of
comparison-based algorithms as the maximal number of comparisons any
individual element takes part in. We give a number of upper and lower
bounds on the \frag for fundamental problems, including
\textsc{Minimum},  \textsc{Selection}, \textsc{Sorting}
and \textsc{Heap Construction}. The results include both deterministic
and randomized upper and lower bounds, and demonstrate a separation
between the two settings for a number of problems. The depth of a
comparator network is a straight-forward upper bound on the worst case
fragile complexity of the corresponding fragile algorithm. We prove that
\frag is a different and strictly easier property than the depth of
comparator networks, in the sense that for some problems a \frag equal
to the best network depth can be achieved with less total work and that
with randomization, even a lower \frag is possible.
\end{abstract} 
\section{Introduction}
\label{sec:intro}

Comparison-based algorithms is a classic and fundamental research area
in computer science. Problems studied include minimum, median, sorting,
searching, dictionaries, and priority queues, to name a few, and by now a
huge body of work exists. The cost measure analyzed is almost always the
total number of comparisons needed to solve the problem, either in the
worst case or the expected case. Surprisingly, very little work has
taken the viewpoint of the individual elements, asking the question: \emph{how
many comparisons must each element be subjected to?}

This question not only seems natural and theoretically fundamental, but
is also practically well motivated: in many real world situations,
comparisons involve some amount of destructive impact on the elements
being compared, hence, keeping the maximum number of comparisons for each
individual element low can be important. One example of such a situation
is ranking of any type of consumable objects (wine, beer, food,
produce), where each comparison reduces the available amount of the
objects compared. Here, classical algorithms like \textsc{QuickSort}, which
takes a single object and partitions the whole set with it, may use up
this pivot element long before the algorithm completes. Another
example is sports, where each comparison constitutes a match and takes a
physical toll on the athletes involved. If a comparison scheme subjects
one contestant to many more matches than others, both fairness to
contestants and quality of result are impacted. The selection process
could even contradict its own purpose---what is the use of finding a
national boxing champion to represent a country at the Olympics if the
person is injured in the process? Notice that in both examples above,
quality of elements is difficult to measure objectively by a numerical
value, hence one has to resort to relative ranking operations between
opponents, i.e., comparisons.  The detrimental
impact of comparisons may also be of less directly physical nature, for
instance if it involves a privacy risk for the elements compared, or if
bias in the comparison process grows each time the same element is used.

\begin{definition}
  We say that a comparison-based algorithm $\A$ has {\em \frag}
  $\f{n}$ if each individual input element participates in at most
  $\f{n}$ comparisons.  We also say that $\A$ has {\em \work} $\w{n}$
  if it performs at most $\w{n}$ comparisons in total.  We say that a
  particular element~$e$ has \frag $f_e(n)$ in~$A$ if $e$ participates
  in at most $f_e(n)$ comparisons.
\end{definition}

In this paper, we initiate the study of algorithms'
\frag---comparison-based complexity from the viewpoint of the
individual elements---and present a number of upper and lower bounds
on the \frag for fundamental problems.

\subsection{Previous work}

One body of work relevant to  what we study here is the
study of sorting networks, propelled by the 1968 paper of
Batcher~\cite{Batc68}. In sorting networks, and more generally
comparator networks (see Section~\ref{sec:defs} for a definition), 
the notions of depth and size correspond to fragile complexity and standard
worst case complexity,\footnote{For clarity, in the rest of the paper we
call standard worst case complexity \emph{work}.} respectively, since a
network with depth~$f(n)$ and size~$w(n)$ easily can be converted into a
comparison-based algorithm with fragile complexity~$f(n)$ and
work~$w(n)$. 

Batcher gave sorting networks with $\Oh(\log^2 n)$ depth and
$\Oh(n\log^2 n)$ size, based on clever variants of the \textsc{MergeSort}
paradigm. A number of later constructions achieve the same
bounds~\cite{Dowd-Perl-Rudolph-Saks/89, Parberry/92, IPL::ParkerP1989,
books/garland/Pratt72}, and for a long time it was an open question
whether better results were possible. In the seminal result in 1983, Ajtai,
Koml{\'o}s, and Szemer{\'e}di~\cite{aks-halvers,aks} answered this
in the affirmative by constructing a sorting network of $\Oh(\log n)$
depth and $\Oh(n\log n)$ size. This construction is quite complex and
involves expander graphs~\cite{HooLinWig06, journals/fttcs/Vadhan12},
which can be viewed as objects encoding pseudorandomness, and which have
many powerful applications in computer science and mathematics. The size
of the constant factors in the asymptotic complexity of the AKS sorting
network prevents it from being practical in any sense. It was later
modified by
others~\cite{Chvatal/92,conf/stoc/Goodrich14,Paterson/90,journals/algorithmica/Seiferas09},
but finding a simple, optimal sorting network, in particular one not based
on expander graphs, remains an open problem. 
Comparator networks for other
problems, such as selection and heap construction have also been studied
~\cite{alekseev:selection-69,brodal:heap-98,Jimbo,Pippenger91,
Yao.merging}. In all these problems the size of the network is super-linear.

As comparator networks of depth $f(n)$ and size $w(n)$ lead to
comparison-based algorithms with $f(n)$ fragile
complexity and $w(n)$ work, a natural question is, whether the
two models are equivalent, or if there are problems for which
comparison-based algorithms can achieve either asymptotically lower $f(n)$, or asymptotically lower $w(n)$ for the same $f(n)$.

One could also ask about the relationship between
parallelism and \frag. We note that parallel time in standard parallel
models generally does not seem to capture \frag. For example, 
even in the most restrictive 
exclusive read and exclusive write (EREW)
PRAM model it is possible to create $n$ copies of an element $e$
in $\Oh(\log n)$ time and, thus, compare $e$ to all the other input
elements in $\Oh(\log n)$ time, resulting in $\Oh(\log n)$ parallel time
but $\Omega(n)$ fragile complexity.  
Consequently, it is not clear
whether Richard Cole's celebrated parallel merge sort
algorithm~\cite{Cole} yields a comparison-based algorithm with low
fragile complexity as it copies some elements.

\def\thmref#1{\small~[T~\ref{thm:#1}]}
\def\lmref#1{\small~[Lem~\ref{lem:#1}]}
\def\cororef#1{\small~[Cor~\ref{cor:#1}]}
\def\obsref#1{\small~[Obs~\ref{obs:#1}]}
\begin{table}[h]
  \def\Deterministic{Determ.}	
  \def\Randomized{Rand.}
	
  \begin{center}
    \hspace*{-.5cm}
    \scalebox{0.87}{
    \begin{tabular}{|l|c|c|c|c|}\hline
    \multicolumn{2}{|l|}{\multirow{2}{*}{\textsl{Problem} } } &\multicolumn{2}{c|}{\textsl{Upper}} & \textsl{Lower} \\
    \cline{3-5}
    \multicolumn{2}{|l|}{}  & $f(n)$ & $w(n)$ & $f(n)$ \\\hline\hline
     
    \multirow{7}{*}{\textsc{Minimum}} & \Deterministic & \multirow{2}{*}{$\Oh(\log n)$~\thmref{min-det-fragile-theta}}                                     & \multirow{2}{*}{$\Oh(n)$}                  & \multirow{2}{*}{$\fmin=\Omega(\log n)$\thmref{min-det-fragile-theta}} \\ 
                                                        &                                   (Sec.~\ref{sec:main-minimum})  & & & \\
      \cline{2-5}  & & & & \\[-.37cm]
 & \Randomized &  $\left\langle \Oh(\log_\Delta n)^\dagger, \Oh(\Delta+\log_\Delta n)^\dagger \right\rangle$\thmref{bestmin}
                                    & & $\langle \Omega(\log_{\Delta} n)^\dagger, \Delta \rangle$\thmref{exp-min-lower} \\
    & (Sec.~\ref{sec:main-minimum}) &   $\left\langle \Oh(1)^\dagger, \Oh(n^\varepsilon) \right\rangle$ (setting $\Delta=n^\varepsilon$)  & $\Oh(n)$ & \\ 

    & & $O(\frac{\log n}{\log\log n})^\dagger$\cororef{thetaallsublog} & & $\Omega(\frac{\log n}{\log\log n})^\dagger$\cororef{thetaallsublog} \\
      \cline{3-5}
 & &  $\hspace*{-2ex}\left\langle  O(\log_\Delta n\log\log \Delta)^\ddagger ,\right.\qquad\qquad$  & $\Oh(n)$ & $\fmin=$ \\
 & & $\quad\left.O(\Delta +\log_\Delta n\log\log \Delta)^\ddagger  \right\rangle $\thmref{bestmin} & $\Oh(n)$ & $=\Omega(\log\log n) ^\ddagger$\thmref{minWHPlb}   \\
 \hline 
     \multirow{5}{*}{\textsc{Selection}} & \Deterministic& \multirow{2}{*}{$\Oh(\log n)$\thmref{detMed}} & \multirow{2}{*}{$\Oh(n)$\thmref{detMed}} & \multirow{2}{*}{$\Omega(\log n)$\cororef{detMedLB}}\\
 & (Sec.~\ref{sec:main-selection})      & & & \\\cline{2-5}
    &\Randomized      & $\left\langle \Oh(\log\log n)^\dagger, \Oh\left(\sqrt n\right)^\dagger\right\rangle$\thmref{RandSelectionCombined} & \multirow{2}{*}{$\Oh(n)^{\dagger}$} & \multirow{2}{*}{$\left\langle \Omega(\log_{\Delta} n)^\dagger, \Delta \right\rangle$\thmref{exp-min-lower} }\\
    &            (Sec.~\ref{sec:main-selection})                                          & $\left\langle \Oh\left(\frac{\log n}{\log \log n}\right)^\dagger, 
                                                                     \Oh(\log^2 n)^\dagger \right\rangle$\thmref{RandSelectionCombined}            &                           & \\\hline
     \multirow{2}{*}{\textsc{Merge}} &\Deterministic       & \multirow{2}{*}{$\Oh(\log n)$\thmref{exp_merging}} & \multirow{2}{*}{$\Oh(n)$}                  & \multirow{2}{*}{$\Omega(\log n)$\lmref{merge}} \\
		   & (Sec.~\ref{sec:main-sorting}) & & & \\\hline
     \textsc{Heap}                  & \Deterministic       & \multirow{2}{*}{$\Oh(\log n)$\obsref{heapconstruction}}                                     & \multirow{2}{*}{$\Oh(n)$}                  & \multirow{2}{*}{$\Omega(\log n)$\thmref{min-det-fragile-theta}}\\
     \textsc{Constr.} & (Sec.~\ref{sec:heapconstruction})    & & & \\\hline 
    \end{tabular}}
  \end{center}
  \caption{  Summary of presented results. Notation: $f(n)$ means \frag;  $w(n)$ means work; $\langle f_m(n), f_{rem}(n) \rangle$ means fragile complexity for the selected element (minimum/median) and for the remaining elements, respectively -- except for lower bounds, where it means $\langle$expected for the selected, limit for remaining$\rangle$;  $\dagger$ means holds in expectation; $\ddagger$ means  holds with high probability ($1-1/n$). %
    $\varepsilon > 0$ is an arbitrary constant.
  }
  \label{tab:result-summary}
\end{table}

\subsection{Our contribution}
In this paper we present algorithms and lower bounds for a number of classical problems, summarized in Table~\ref{tab:result-summary}.
In particular, we study finding the \textsc{Minimum} (Section~\ref{sec:main-minimum}), the \textsc{Selection} problem (Section~\ref{sec:main-selection}), 
and \textsc{Sorting} (Section~\ref{sec:main-sorting}).
\subparagraph{Minimum.}
The case of the deterministic algorithms is clear: using an adversary lower bound, we show that the minimum element needs to suffer
$\Omega(\log n)$ comparisons and a tournament tree trivially achieves this bound (Subsection~\ref{sec:min-deterministic}).
The randomized case, however, is much more interesting. 
We obtain a simple algorithm where the probability of the minimum element suffering $k$ comparisons is doubly exponentially low in $k$, roughly $1/2^{2^k}$
(see Subsection~\ref{sec:randAlgMin}). 
As a result, the $\Theta(\log n)$ deterministic \frag can be lowered to $O(1)$ expected or even $O(\log\log n)$ with high probability.
We also show this latter high probability case is lower bounded by $\Omega(\log\log n)$ (Subsection~\ref{sec:min-lower}).
Furthermore, we can achieve a trade-off between the \frag of the minimum element and the other elements.
Here $\Delta=\Delta(n)$ is a parameter we can choose freely that basically upper bounds the \frag of the non-minimal elements.
We can find the minimum with $O(\log_\Delta n)$ expected  \frag while all the other elements suffer $O(\Delta + \log_\Delta n)$ comparisons
(Subsection~\ref{sec:min-lower}). 
Furthermore, this is tight: we show an $\Omega(\log_\Delta n)$ lower bound for the expected \frag of the minimum element where
the maximum \frag of non-minimum elements is at most $\Delta$.

\subparagraph{Selection.} Minimum finding is a special case of the
selection problem where we are interested in finding an element of a
given rank.  As a result, all of our lower bounds apply to this problem
as well.  Regarding upper bounds, the deterministic case is trivial if
we allow for $O(n\log n)$ work (via sorting).  We show that this can be
reduced to $O(n)$ time while keeping the \frag of all the elements at
$O(\log n)$ (Section~\ref{sec:main-selection}).  Once again,
randomization offers a substantial improvement: e.g., we can find the
median in $O(n)$ expected work and with $O(\log\log n)$ expected \frag
while non-median elements suffer $O(\sqrt{n})$ expected comparisons, or
we can find the median in $O(n)$ expected work and with
$O(\log n / \log\log n)$ expected \frag while non-median elements suffer
$O(\log^2{n})$ expected comparisons.
 
\subparagraph{Sorting and other results.}  The deterministic
selection, sorting, and heap construction \frags follow directly from
the classical results in comparator
networks~\cite{aks,brodal:heap-98}.  However, we show a separation
between comparator networks and comparison-based algorithms for the
problem of \textsc{Median} (Section~\ref{sec:main-selection}) and
\textsc{Heap Construction} (Section~\ref{sec:heapconstruction}), in
the sense that depth/\frag of $\Oh(\log n)$ can be achieved in
$\Oh(n)$ work for comparison-based algorithms, but requires
$\Omega(n \log n)$~\cite{alekseev:selection-69} and
$\Omega(n \log\log n)$~\cite{brodal:heap-98} sizes for comparator
networks for the two problems, respectively.
For sorting the two models achieve the same complexities: $\Oh(\log n)$
depth/\frag and $\Oh(n \log n)$ size/work, which are the
optimal bounds in both models due to the  $\Omega(\log n)$
lower bound on \frag for \textsc{Minimum} (Theorem~\ref{thm:min-det-fragile-theta}) and the standard
$\Omega(n \log n)$ lower bound on work for comparison-based sorting. However, it is an
open problem whether these bounds can be achieved by simpler sorting
algorithms than sorting networks, in particular whether expander
graphs are necessary. One intriguing conjecture could be that any
comparison-based sorting algorithm with $\Oh(\log n)$ \frag and
$\Oh(n \log n)$ work implies an expander graph. This would imply
expanders, optimal sorting networks and fragile-optimal comparison-based
sorting algorithms to be equivalent, in the sense that they all encode
the same level of pseudorandomness.

We note that our lower bound of $\Omega(\log^2 n)$ on the \frag of \textsc{MergeSort}
(Theorem~\ref{thm:mergesort-lower}) implies the same lower bound on the depth of any sorting network based
on binary merging, which explains why many of the existing simple
sorting networks have $\Theta(\log^2 n)$ depth. Finally, our analysis of
\textsc{MergeSort} on random inputs (Theorem~\ref{thm:mergesort-whp}) shows a separation between deterministic and randomized
\frag for such algorithms.
In summary, we consider the main contributions of this paper to be: 
\begin{itemize}
	\item the introduction of the model of \frag, which we find intrinsically
	interesting, practically relevant, and surprisingly overlooked
	\item the
	separations between this model and the model of comparator networks
	\item the separations between the deterministic and randomized setting within
	the model
	\item the lower bounds on randomized minimum finding
\end{itemize}
\section{Definitions}
\label{sec:defs}

\begin{figure}
  \begin{center}
    \scalebox{0.9}{%
\begin{tikzpicture}[
    ->-/.style={decoration={markings, mark=at position .98 with {\arrow{latex}}},postaction={decorate}}  
]

  \foreach \y in {1, ..., 8} {
    \node[anchor=east] at (0, -\y em) {$x_\y$};
    \path[draw] (0, -\y em) to (24em, -\y em);
  }
  \newcommand{\comparator}[4]{%
    \def\posx{#3*3em}
    \draw[fill=black] (\posx, -#1 em) circle (0.1em);
    \draw[fill=black] (\posx, -#2 em) circle (0.15em);
    \path[draw, black, ->-, #4] (\posx, -#2 em) to (\posx, -#1 em);
  }

  \comparator{1}{2}{1}{}
  \comparator{3}{4}{1}{}
  \comparator{5}{6}{1}{}
  \comparator{7}{8}{1}{}
  
  \comparator{1}{3}{1.9}{}
  \comparator{2}{4}{2.1}{}
  \comparator{5}{7}{1.9}{}
  \comparator{6}{8}{2.1}{}

  \comparator{2}{3}{3}{}
  \comparator{6}{7}{3}{}

  \comparator{1}{5}{4.7}{}
  \comparator{2}{6}{4.9}{}
  \comparator{3}{7}{5.1}{}
  \comparator{4}{8}{5.3}{}

  \comparator{3}{5}{5.9}{}
  \comparator{4}{6}{6.1}{}

  \comparator{2}{3}{7}{}
  \comparator{4}{5}{7}{}
  \comparator{6}{7}{7}{}
\end{tikzpicture}
}
    \vspace{-1em}
  \end{center}
  
  \caption{
    Batcher's Odd-Even-Mergesort network~\cite{Batc68}: $8$~inputs, depth $f(8){=}6$ and size $w(8){=}19$.
  }
  \label{fig:sorting-network}
\end{figure}
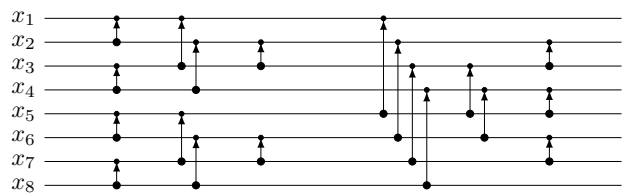

\textbf{Comparator networks.\ }
A comparator network $\network$ is constructed of {\em comparators} each consisting of two inputs and two (ordered) outputs.
The value of the first output is the minimum of the two inputs and the value of the second output is the maximum of the two inputs.
By this definition a comparator network $\network_n$ on $n$ inputs also consists of $n$ outputs.
Figure~\ref{fig:sorting-network} demonstrates a common visualization of comparator networks with values as horizontal wires and comparators represented by vertical arrows between pairs of wires.
Each arrow points to the output that returns the minimum input
value. Inputs are on the left and outputs on the right.
The {\em size} of the comparator network is defined as the number of
comparators in it, while its {\em depth} is defined as the number of
comparators on the longest path from an input to an output.

We note that a comparator network is straightforward to execute as a
comparison-based algorithm by simulating its comparators sequentially
from left to right (see Figure~\ref{fig:sorting-network}), breaking
ties arbitrarily. If the network has depth~$f(n)$ and size~$g(n)$, the
comparison-based algorithm clearly has fragile complexity~$f(n)$ and
work~$g(n)$.

\textbf{Networks for problems.\ }
We define the set of inputs to the comparator network $\network$ by $\I$ and the
outputs by $\network(\I)$.  We use the notation $\network(\I)^{i}$ for
$0\le i\le n-1$, to represent the $i$-th output and $\network(\I)^{i:j}$
for $0 \le i \le j \le n-1$ to represent the ordered subset of the
$i$-th through $j$-th outputs.

An {\em \sorting network} is a comparator network $\network_n$ such that
$\network_n(\I)^t$ carries the $t$-th smallest input value for all~$t$. We say such a network solves the {\em \sorting problem}.

An {\em \selection{t} network} is a comparator network $\network_{n,t}$ such that $\network_{n,t}(\I)^0$ carries the $t$-th smallest input value. We say such a network solves the {\em \selection{t} problem}.

An {\em \partition{t} network} is a comparator network $\network_{n,t}$, such that $\network_{n,t}(\I)^{0:t-1}$ carry the $t$ smallest input values.\footnote{Brodal and Pinotti~\cite{brodal:heap-98} call it an \selection{t} network, but we feel \partition{t} network is a more appropriate name.}  We say such a network solves the {\em \partition{t} problem}.

Clearly, an \selection{t} problem is asymptotically no harder
than \partition{t} problem: let $\network^{\downarrow}_{n,t}(\I)$ denote
an \partition{t} network with all comparators reversed; then
$\network'_{n,t}(\I) =
\network^{\downarrow}_{t,t-1}(\network_{n,t}(\I))$
is an \selection{t} network. However, the converse is not clear: given a
value of the $t$-th smallest element as one of the inputs, it is not
obvious how to construct an \partition{t} network with smaller size or
depth.  In Section~\ref{sec:main-selection} we will show that the two
problems are equivalent: every \selection{t} network also solves
the \partition{t} problem.

\textbf{Rank.\ }
Given a set $X$, the rank of some element $e$ in $X$, denoted by $\rankx{e}$, is equal to the
size of the subset of $X$ containing the elements that are no larger
than $e$. When the set $X$ is clear from the context, we will omit the subscript $X$ and simply write $\rank{e}$.
 
\section{Finding the minimum}\label{sec:main-minimum}
\subsection{Deterministic Algorithms}\label{sec:min-deterministic}
As a starting point, we study deterministic algorithms that find the
minimum among an input of $n$ elements.  Our results here are simple but
they act as interesting points of comparison against the
subsequent non-trivial results on randomized algorithms.

\begin{theorem}\label{thm:min-det-fragile-theta}
  The fragile complexity of finding the minimum of $n$ elements is $\lceil \log n \rceil$.
\end{theorem}
\begin{proof}
  The upper bound follows trivially from the application of a balanced tournament tree.
  We thus focus on the lower bound.
  Let $S$ be a set of $n$ elements, and $\A$ be a deterministic comparison-based algorithm that finds the minimum element of $S$. We describe an adversarial strategy for resolving the comparisons made by $\A$ that leads to the lower bound.

  Consider a directed graph $G$ on $n$ nodes corresponding to the elements of $S$.
  With a slight abuse of notation, we use the same names for elements of $S$ and the associated nodes in graph $G$. 
  The edges of $G$ correspond to comparisons made by $\A$, and are either black or red.
  Initially $G$ has no edges. 
  If $\A$ compares two elements, we insert a directed edge between the associated nodes pointing toward the element declared smaller by the adversarial strategy. 
  Algorithm $\A$ correctly finds the minimum element if and only if, upon termination of~$\A$, the resulting graph $G$ has only one sink node.

  Consider the following adversarial strategy to resolve comparisons made by $\A$: if both elements are sinks in~$G$, the element that has already participated in more comparisons is declared smaller; if only one element is a sink in~$G$, this element is declared smaller; and if neither element is a sink in~$G$, the comparison is resolved arbitrarily (while conforming to the existing partial order).

  We color an edge in $G$ red if it corresponds to a comparison between two sinks; otherwise, we color the edge black. For each element $i$, consider its in-degree $d_i$ and the number of nodes $r_i$ in $G$ (incl. $i$ itself) from which $i$ is reachable by a directed path of only red edges.

  We show by induction that $r_i \leq 2^{d_i}$ for all sinks in $G$. 
  Initially, $r_i=1\leq 1=2^{d_i}$ for all~$i$.
  Let algorithm $\A$ compare two elements $i$ and $j$, where $i$ is a sink, and let the adversarial strategy declare $i$ to be smaller than $j$.
  Then, the resulting in-degree of $i$ is $d_i + 1$.
  If the new edge is black, the number of nodes from which $i$ is reachable via red edges does not change, and the inequality holds trivially.
  If the new edge is red, the resulting number of nodes from which $i$ is reachable is 
  $r_i+r_j \leq 2^{d_i}+2^{d_j} \leq 2^{d_i+1}$. 
  Therefore, when $\A$ terminates with the only sink $v$ in $G$, which represents the minimum element, with degree $d_v \ge \lceil\log r_v\rceil$. %
  The theorem follows by observing that a tournament tree is an instance where $r_v = n$.
\end{proof}
Observe that in addition to returning the minimum, the balanced tournament tree can also return the second smallest element, without any increase to the \frag of the minimum. We refer to this deterministic algorithm that returns the smallest and the second smallest element of a set $X$ as \Call{TournamentMinimum}{$X$}.

\begin{corollary}\label{cor:detMedLB}
  For any deterministic algorithm $\A$ that finds the median of~$n$ elements, the \frag of the median element is at least $\lceil \log n \rceil -1$.
\end{corollary}
\begin{proof}
  By a standard padding argument with $n-1$ small elements.
\end{proof}
 \subsection{Randomized Algorithms for Finding the Minimum}\label{sec:randAlgMin}
We now show that finding the minimum is provably easier for randomized
algorithms than for deterministic algorithms. 
We define $\fmin$ as the \frag of the minimum and
$\frem$ as the maximum \frag of the remaining elements. 
For deterministic algorithms we have shown that $\fmin \ge \log n$ regardless of
$\frem$.
This is very different in the randomized setting. 
In particular, we first show that we can achieve $\E{\fmin} = O(1)$ and
$\fmin = O(1) + \log\log n$ whp. (in Theorem~\ref{thm:minWHPlb} we show that this bound is also tight).

\vspace{0.5em}
\noindent\fbox{\scalebox{0.95}{\begin{minipage}{\textwidth}
\begin{algorithmic}[1]
  \Procedure{SampleMinimum}{$X$}   \Comment{Returns the smallest and 2nd smallest element of $X$}
  \If{$|X| \leq 8$} \Return \Call{TournamentMinimum}{$X$}
  \EndIf
  \State Let $A\subset X$ be a uniform random sample of $X$, with $|A|=\ceil{|X|/2}$ \label{cmbMnSmpA}
  \State Let $B\subset A$ be a uniform random sample of $A$, with $|B|=\floor{|X|^{2/3}}$ \label{cmbMnSmpB}\\
  \Comment{The minimum is either in (i) $C\subseteq X\setminus A$, (ii) $D \subseteq A\setminus B$ or (iii) $B$} 
  \State $(b_1,b_2)=$ \Call{SampleMinimum}{$B$} \label{cmbMnRecB} \Comment{the minimum participates only in case (iii)}
  \State Let $D= \{ x \in A \setminus B \mid x < b_2\}$  \label{cmbMnFilterB} \Comment{the minimum is compared once only in case (ii)}
  \State Let $(a'_1,a'_2)=$ \Call{SampleMinimum}{$D$} \label{cmbMnRecA} \Comment{only case (ii)}
  \State Let $(a_1,a_2)=$ \Call{TournamentMinimum}{$a'_1,a'_2,b_1,b_2$} \label{cmbMnBD} \Comment{case (ii) and (iii)}
  \State Let $C= \{ x \in X \setminus A \mid x < a_2\}$ \label{cmbMnFilterA} \Comment{only case (i)}
  \State Let $(c_1,c_2)=$ \Call{TournamentMinimum}{$C$} \label{cmbMnRecC} \Comment{only case (i)}
  \State \Return \Call{TournamentMinimum}{$ a_1,a_2,c_1,c_2$} \label{cmbMnRet} \Comment{always}
\EndProcedure
\end{algorithmic}
\end{minipage}}
}

\vspace{0.5em}

First, we show that this algorithm can actually find the minimum with expected constant
number of comparisons.
Later, we show that the probability that this algorithm performs
$t$ comparisons on the minimum drops roughly \textit{doubly exponentially} on $t$. 

We start with the simple worst-case analysis.
\begin{lemma}\label{lem:log}
  Algorithm \Call{SampleMinimum}{$X$} achieves $\fmin \leq 3\log |X|$ in the worst case.
\end{lemma}
\begin{proof}
  First, observe that the smallest element in Lines~\ref{cmbMnBD} and~\ref{cmbMnRet} participates in at most one comparison because pairs of elements are already sorted. Then the \frag of the minimum is defined by the maximum of the three cases:
  \begin{enumerate}[(i)]
    \item One comparison each in Lines~\ref{cmbMnFilterA} and \ref{cmbMnRet}, plus (by Theorem~\ref{thm:min-det-fragile-theta}) $\ceil{\log |C|} \le \log|X|$ comparisons in Line~\ref{cmbMnRecC}. 
    \item One comparison each in Lines~\ref{cmbMnFilterB}, \ref{cmbMnBD}, and \ref{cmbMnRet}, plus the recursive call in line~\ref{cmbMnRecA}. 
    \item One comparison each in Lines~\ref{cmbMnRecB}, \ref{cmbMnBD}, and \ref{cmbMnRet}, plus the recursive call in line~\ref{cmbMnRecB}. 
  \end{enumerate}

  The recursive calls in lines~\ref{cmbMnRecA} and~\ref{cmbMnRecB} are on at most $|X|/2$ elements because $B \subset A$, $D \subset A$, and $|A| = \ceil{|X|/2}$. 
  Consequently, the \frag of the minimum is governed by 
\begin{equation*}
T(n) \le \left\{\begin{array}{ll}
		   \max\left\{ 3+T(n/2),  2+ \log n\right\} & \mbox{ if } n > 8  \\
		   3 & \mbox{ if } n \le 8
		\end{array} \right.,
\end{equation*}
which solves to $T(n) \le 3\log n$.
\end{proof}

\begin{lemma}\label{lem:distrC}
  Assume that in Algorithm \Call{SampleMinimum}{}, the minimum~$y$ is in $X\setminus A$, i.e. we are in case~(i).
  Then $\p[|C|=k\mid y \not\in A] \leq \frac{k}{2^{k}}$ for any $k\ge 1$ and $n\geq 7$. 
\end{lemma}
\begin{proof}
 There are $\binom{n-1}{\ceil{n/2}}$ possible events of choosing a random subset $A \subset X$ of size $\ceil{n/2}$ s.t.\ $y \not \in A$. Let us count the number of the events $\{|C|=k \mid y \not \in A\}$, which is 
equivalent to $a_2$, the second smallest element of $A$, being larger than exactly  $k+1$ elements of $X$. %

For simplicity of exposition, consider the elements of $X = \{x_1, \dots, x_n\}$ in sorted order. The minimum $y=x_1 \not \in A$, therefore, $a_1$ (the smallest element of $A$) must be one of the $k$ elements $\{x_2, \dots, x_{k+1}\}$. By the above observation, $a_2 = x_{k+2}$. And the remaining $\ceil{n/2}-2$ elements of $A$ are chosen from among $\{x_{k+3}, \dots, x_n\}$. Therefore,

\begin{equation*}
  \p[|C|=k\mid y \not\in A] = \frac{k \cdot \binom{n-(k+2)}{\ceil{n/2}-2}}{\binom{n-1}{\ceil{n/2}}} 
  = k \cdot \frac{(n-(k+2))!}{(\floor{n/2}-k)! (\ceil{n/2}-2)!} \cdot \frac{(\ceil{n/2})!(\floor{n/2}-1)!}{(n-1)!}
\end{equation*}

Rearranging the terms, we get:

\begin{align*}
  \p[|C|=k\mid y \not\in A] 
  &= k \cdot \frac{(n-(k+2))!}{(n-1)!} \cdot \frac{(\ceil{n/2})!}{(\ceil{n/2}-2)!}\cdot \frac{(\floor{n/2}-1)!}{(\floor{n/2}-k)!} 
\end{align*}
There are two cases to consider:

\begin{align*}
  k = 1: \mbox{\quad\quad} && \p[|C|=k\mid y \not\in A] & = 1 \cdot \frac{1}{(n-1)(n-2)} \cdot  \ceil{n/2}\left(\ceil{n/2}-1\right) \cdot 1 \\
  &&&\le \frac{1}{(n-1)(n-2)} \cdot  \frac{(n+1)}{2}\cdot \frac{(n-1)}{2} \\
  &&&=\frac{n+1}{4 \cdot (n-2)} \le \frac{1}{2} = \frac{k}{2^k} \mbox{ \hspace{4em} for every } n \ge 5.
\end{align*} \\

\begin{align*}
  k \ge 2: \mbox{\quad\quad} && \p[|C|=k\mid y \not\in A] &= k \cdot \frac{1}{\prod_{i=1}^{k+1} (n-i)} \cdot \ceil{n/2}\left(\ceil{n/2}-1\right) \cdot \prod_{i=1}^{k-1} \left(\floor{\frac{n}{2}}-i\right)\\
  &&&\le k \cdot \frac{1}{\prod_{i=1}^{k+1} (n-i)} \cdot \frac{n+1}{2}\cdot \frac{n-1}{2} \cdot \prod_{i=1}^{k-1} \frac{n-2i}{2} \\
  &&&\le \frac{k}{2^{k+1}} \cdot (n+1)(n-1) \cdot \frac{\prod_{i=1}^{k-1} (n-2i)}{\prod_{i=1}^{k+1} (n-i)} \\
  &&&\le		\frac{k}{2^{k+1}} \cdot (n+1)(n-1)\cdot \frac{n-2}{(n-1)(n-2)(n-3)} \\
  &&&= \frac{k}{2^{k+1}} \cdot \frac{n+1}{n-3} \le \frac{k}{2^{k+1}}\cdot 2 = \frac{k}{2^k} \mbox{ \hspace{1em}  for every } n\ge 7. \qedhere
\end{align*}
\end{proof}

\def\cmbMnE{9}
\begin{theorem}\label{thm:expectedConstant}
  Algorithm \Call{SampleMinimum}{} achieves $\E{\fmin} \leq \cmbMnE$.
\end{theorem}
\begin{proof}
  By induction on the size of $X$.
  In the base case $|X|\leq 8$, clearly $\fmin \leq 3$, implying the theorem.
  Now assume that the calls in Line~\ref{cmbMnRecA} and Line~\ref{cmbMnRecB} have the property that $\E{f(b_1)}\leq \cmbMnE$ and $\E{f(a'_1)}\leq \cmbMnE$.
  Both in case (ii) and case (iii), the expected number of comparisons of the minimum is $\leq \cmbMnE + 3$.
  Case (i) happens with probability at least 1/2.
  In this case, the expected number of comparisons is 2 plus the ones from Line~\ref{cmbMnRecC}.
  By Lemma~\ref{lem:distrC} we have $\p[|C|=k\mid \hbox{case (i)}] \leq k2^{-k}$.
  Because \Call{TournamentMinimum}{} (actually any algorithm not repeating the same comparison) uses the minimum at most~$k-1$ times, the expected number of comparisons in Line~\ref{cmbMnRecC} is $\sum_{k=1}^\floor{n/2} (k-1)k2^{-k} \leq \sum_{k=1}^\infty (k-1)k2^{-k} \leq 4$.
  Combining the bounds we get $\E{\fmin} \leq \frac{\cmbMnE +3}{2}+\frac{2+4}{2} = \cmbMnE$.
\end{proof}

Observe that the above proof did not use anything about the sampling of~$B$, and also did not rely on \Call{TournamentMinimum}{}.

\begin{lemma}\label{lem:expSample}
  For $|X| > 2$ and any $\gamma>1$: $\p\left[ |D| \ge \gamma |X|^{1/3} \right] < |X|\exp(-\Theta(\gamma))$
\end{lemma}
\begin{proof}
  Let $n=|X|$, $a = |A| = \ceil{n/2}$ and $b=|B|=\floor{n^{2/3}}$.
  The construction of the set $B$ can be viewed as the following experiment. Consider drawing without replacement from an urn with $b$ blue and $a-b$ red marbles. The $i$-th smallest element of $A$ is chosen into $B$ iff the $i$-th draw from the urn results in a blue marble. 
  Then $|D| \ge \gamma |X|^{1/3} = \gamma n^{1/3}$ implies that this experiment results in at most one blue marble among the first $t = \gamma n^{1/3}$ draws.
  There are precisely $t+1$ elementary events that make up the condition $|D|\ge t$, namely that the $i$-th draw is a blue marble, and where $i=0$ stands for the event ``all $t$ marbles are red''.
  Let us denote the probabilities of these elementary events as $p_i$.

  Observe that each $p_i$ can be expressed as a product of $t$ factors, at least $t-1$ of which stand for drawing a red marble, each upper bounded by $1-\frac{b-1}{a}$.
  The remaining factor stands for drawing the first blue marble (from the urn with $a-i$ marbles, $b$ of which are blue), or another red marble.
  In any case we can bound
  \[
     p_i \le \left(1-  {\frac {b-1}a}\right)^{t-1} \le \left(1-  {\frac {b-1}a}\right)^{\gamma n^{1/3}-1} = \exp\left(-\Theta\left(\frac{b\gamma n^{1/3}}{a}\right)\right).
  \]
  
  Summing the $t+1$ terms, and observing $t+1<n$ if the event can happen at all, we get
  \begin{equation*}
    \p[|D| \ge \gamma |X|^{1/3}] <  n\cdot \exp\left(-\Theta\left(\frac{\gamma n^{1/3}n^{2/3}}{n/2}\right)\right) = n\cdot \exp\left( -\Theta(\gamma) \right). \qedhere
  \end{equation*}
\end{proof}

\begin{theorem}\label{thm:loglogWHP}
  There is a positive constant $c$, such that for any parameter $t \ge c$, 
  the minimum in the Algorithm \Call{SampleMinimum}{$X$} participates in at most $O(t + \log \log |X|)$ comparisons 
  with probability at least $1-\exp(-2^t)2\log\log |X|$.
\end{theorem}
\begin{proof}
  Let  $n = |X|$ and $y$ be the minimum element.
  In each recursion step, we have one of three cases:
  (i) $y\in C\subseteq X\setminus A$, (ii) $y\in D \subseteq A\setminus B$ or (iii) $y\in B$.
  Since the three sets are disjoint, the minimum always participates in at most one recursive call.
  Tracing only the recursive calls that include the minimum, we 
  use the superscript $X\tin{i}, A\tin{i}, B\tin{i}, C\tin{i}$, and $D\tin{i}$ to denote these sets 
  at depth $i$ of the recursion. 
  
  Let $h$ be the first recursive level when $y \in C\tin{h}$, i.e., $y \not \in A\tin{h}$.
  It follows that $y$ will not be involved in the future recursive calls because it is in a single call to \Call{TournamentMinimum}{}.
  Thus, at this level of recursion, the number of comparisons that $y$ will accumulate is 
  equal to $O(1)+\log|C\tin{h}|$.
  To bound this quantity, let $k = 4\cdot 2^t$. Then, by Lemma~\ref{lem:distrC}, $\p[|C\tin{h}|> k]\leq k2^{-k} = 4 \cdot 2^t \cdot 2^{-4\cdot 2^t} = 4 \cdot 2^t \cdot 4^{-2^t} \cdot 4^{-2^t}$. Since $4x4^{-x} \le 1$ for any $x\ge 1$,  $\p[|C\tin{h}|> k]\leq 4^{-2^t}$ for any $t \ge 0$. I.e., the number of comparisons that $y$ participates in at level $h$ is at most $O(1) + \log k = O(1) + t$ with probability at least $1-4^{-2^t} \ge 1-\exp(-2^t)$.

  Thus, it remains to bound the number of comparisons involving $y$ at the recursive levels $i \in [1, h-1]$.
  In each of these recursive levels $y \not \in C\tin{i}$, which only leaves the two cases: (ii) $y \in D\tin{i} \subseteq A\tin{i}\setminus B\tin{i}$
  and (iii) $y \in B\tin{i}$.
  The element $y$ is involved in at most $O(1)$ comparisons in lines \ref{cmbMnFilterB}, \ref{cmbMnBD} and \ref{cmbMnRet}.
  The two remaining lines of the algorithm are lines \ref{cmbMnRecB} and \ref{cmbMnRecA} which are the recursive calls. 
  We differentiate two types of recursive calls: 
  \begin{itemize}
      \item Type 1: $|X\tin{i}| \le 2^{4t}$. 
          In this case, by Lemma~\ref{lem:log}, the algorithm will perform 
          $O(t)$ comparisons at the recursive level $i$, as well as any subsequent recursive levels. 

      \item Type 2: $|X\tin{i}| > 2^{4t}$. 
          In this case, by Lemma~\ref{lem:expSample} on the set $X\tin{i}$ and $\gamma = |X\tin{i}|^{1/3}$ we get:  
          \[
              \p[|D\tin{i}| \ge \gamma |X\tin{i}|^{1/3}]  < |X\tin{i}| \exp\left(-\Theta\left(|X\tin{i}|^{1/3}\right)\right) <\exp\left(-\Theta\left(|X\tin{i}|^{1/3}\right)\right)
          \]

	  Note that since $|X\tin{i}|^{1/3} > 2^{t}$, by the definition of the $\Theta$-notation, there exists a positive constant $c$, such that $\exp\left(-\Theta\left(|X\tin{i}|^{1/3}\right)\right) < \exp(-2^t)$.
          Thus, it follows that with probability $1-\exp(-2^t)$, we will recurse on a subproblem of
          size at most $\gamma|X\tin{i}|^{1/3} \le |X\tin{i}|^{2/3}$.
          Let $G_i$ be this (good) event, and thus $\p[G_i] \ge  1-\exp(-2^t)$.
  \end{itemize}
  Observe that the maximum number of times we can have good events of type 2 is very limited.
  With every such good event, the size of the subproblem decreases significantly and thus eventually we will arrive at 
   a recursive call of type 1. 
  Let $j$ be this maximum number of ``good'' recursive levels of type 2.
  The problem size at the  $j$-th such recursive level is
  at most $n^{\left( 2/3 \right)^{j-1}}$ and we must have that 
  $n^{\left( 2/3 \right)^{j-1}} > 2^{4t}$ which reveals that we must have 
  $j = O\left( \log\log n \right)$.

  We are now almost done and we just need to use a union bound.
  Let $G$ be the event that at the recursive level $h$, we perform at most $O(1) + t$ comparisons,
  and all the recursive levels of type 2 are good.
  $G$ is the conjunction of at most $j+1$ events and as we have shown, each such event holds with probability
  at least $1-\exp(-2^t)$.
  Thus, it follows that $G$ happens with probability $1-(j+1)\exp(-2^t) > 1 - 2\log\log n\exp(-2^t)$.
  Furthermore, our arguments show that if $G$ happens, then the minimum will only 
  particpate in $O(t + j) = O\left( t+ \log\log n \right)$ comparisons.
 \end{proof}

The major strengths of the above algorithm is the doubly exponential drop in probability of 
comparing the minimum with too many elements. 
Based on it, we can design another simple algorithm to provide a smooth 
trade-off between $\fmin$ and $\frem$.
Let $2 \le \Delta \le n$ be an integral parameter.
We will design an algorithm that achieves
$\E{\fmin} = O( \log_\Delta n  )$ and 
$\fmin= O( \log_\Delta n \cdot \log\log \Delta )$ whp,
and $\frem =\Delta + O(\log_\Delta n \cdot \log\log \Delta)$ whp.
For simplicity we assume $n$ is a power of $\Delta$.
We build a fixed tournament tree $T$ of degree $\Delta$ and of height $\log_\Delta n$ on $X$.
For a node $v \in T$, let $X(v)$ be the set of values in the subtree rooted at $v$.
The following code computes $m(v)$, the minimum value of $X(v)$, for every node $v$.

\vspace{0.5em}
\noindent{\fbox{\begin{minipage}{0.98\textwidth}
\begin{algorithmic}[1]
  \Procedure{TreeMinimum$_\Delta$}{$X$} 
  \State For every leaf $v$, set $m(v)$ equal to the single element of $X(v)$.
  \State For every internal node $v$ with $\Delta$ children $u_1, \dots, u_\Delta$ where the values
  $m(u_1), \dots, m(u_\Delta)$ are known, compute $m(v)$ using {\sc SimpleMinimum} algorithm on input
  $\left\{ m(u_1), \dots, m(u_\Delta) \right\}$.
    \State Repeat the above step until the minimum of $X$ is computed.
\EndProcedure
\end{algorithmic}
\end{minipage}}}
\vspace{0.5em}

\noindent The correctness of {\sc TreeMinimum$_\Delta$} is trivial.
So it remains to analyze its \frag.

\begin{theorem}\label{thm:bestmin}
  In {\sc TreeMinimum$_\Delta$}, $\E{\fmin} = O(\log_\Delta n)$ and $\E{\frem}= \Delta + O(\log_\Delta n)$.
  Furthermore, with high probability, $\fmin = O\left(\frac{\log n\log\log \Delta}{\log \Delta}\right)$ and 
  $\frem = O\left(\Delta + \frac{\log n\log\log \Delta}{\log \Delta}\right)$.
\end{theorem}
\begin{proof}
    First, observe that $\E{\fmin } = O(\log_\Delta n)$ is an easy consequence of Theorem~\ref{thm:expectedConstant}.
    Now we focus on high probability bounds.
    Let $k=c\cdot h\log\ln \Delta$, and $h=\log_\Delta n$ for a large enough constant $c$.
    There are $h$ levels in $T$.
    Let $\bF_i$ be the random variable that counts the number of comparisons
    the minimum participates in at level $i$ of $T$. 
    Observe that these are independent random variables. 
    Let $f_1, \dots, f_h$ be integers such that $f_i\ge 1$ and
    $\sum_{i=1}^h f_i = k$, and let $c'$ be the constant hidden in the big-$O$ notation of Theorem~\ref{thm:loglogWHP}.
    Use Theorem~\ref{thm:loglogWHP}  $h$ times (with $n$ set to $\Delta$, and $t=f_i$), and also bound $2\log\log \Delta < \Delta$ to get
    \begin{align*}
        \p\Big[\vphantom{\frac12}\bF_1 \geq c'(f_1+\log\log \Delta) \vee \dots \vee \bF_h \geq c'(f_h + \log\log\Delta)\Big] \le \Delta^h e^{-\sum_{i}2^{f_i}} \le \Delta^h e^{-h 2^{k/h}}
 \end{align*}
 where the last inequality follows from the inequality of arithmetic and geometric means
 (specifically, observe that  $ \sum_{i=1}^h 2^{f_i}$ is minimized when all $f_i$'s are distributed
 evenly). %

 Now observe that the total number of different integral sequences $f_1, \dots, f_h$ that sum up to $k$
 is bounded by $h+k \choose h$ (this is the classical problem of distributing $k$ identical balls into $h$  distinct bins).
 Thus, we have
 \begin{multline*}
   \p[\fmin = O(k+h\log\log \Delta)]\le {h+k \choose h} \cdot \Delta^{h} \frac{1}{ e^{h\cdot 2^{k/h }} } \le
   \left(\frac{e(h+k)}{h}\right)^h \cdot \Delta^{h} \frac{1}{ e^{h\cdot 2^{k/h } }}  \\ \le
   \left( \frac{O\left(\frac{k}{h}\right)\cdot \Delta}{ e^{2^{k/h } }}\right)^h =
   \left( \frac{O\left(\Delta^2\right)}{ e^{2^{c\log\ln \Delta } }}\right)^h <
   \left( \frac{O(\Delta^2)}{ e^{\ln^{c}\Delta }}\right)^h <
   \left( \frac{\Delta^3}{ \Delta^{\ln^{c-1}\Delta }}\right)^h<
   \Delta^{-ch}  =
   n^{-c}
 \end{multline*}
 where in the last step we bound $(\ln \Delta)^{c-1}-3 > c$ for large enough $c$ and $\Delta \ge 3$.
    This is a high probability bound for $\fmin$.
    To bound $\frem$, observe that for every non-minimum element $x$, there exists a lowest
    node $v$ such that $x$ is not $m(v)$.
    If $x$ is not passed to the ancestors of $v$, $x$ suffers at most $\Delta$ comparisons in $v$, and
    below $v$ $x$ behaves like the minimum element, which means that the above analysis applies. 
    This yields that whp we have
    $\frem = \Delta + O\left( \frac{\log n\log\log \Delta}{\log \Delta}\right)$.
\end{proof}
\subsection{Randomized Lower Bounds for Finding the Minimum}\label{sec:min-lower}
\subsubsection{Expected Lower Bound for the Fragile Complexity of the Minimum.}
The following theorem is our main result.
\begin{theorem}\label{thm:exp-min-lower}
  In any randomized minimum finding algorithm with \frag of at most $\Delta$ for any element, the expected \frag of the minimum is at least $\Omega(\log_\Delta n)$. 
\end{theorem}

Note that this theorem implies the \frag of finding the minimum:
\begin{corollary}\label{cor:thetaallsublog}
  Let $f(n)$ be the expected \frag of finding the minimum (i.e. the smallest function such that some algorithm achieves $f(n)$ \frag for all elements (minimum and the rest) in expectation).
  Then $f(n) = \Theta (\frac{\log n}{\log\log n})$.
\end{corollary}
\begin{proof}
  Use Theorem~\ref{thm:bestmin} as the upper bound and 
  Theorem~\ref{thm:exp-min-lower}, both with $\Delta = \frac{\log n}{\log\log n}$, observing that if $f(n)$ is an upper bound that holds with high probability, it is also an upper bound on the expectation.
\end{proof}

To prove Theorem~\ref{thm:exp-min-lower} we give a lower bound for a {\em deterministic} algorithm~$\A$ on a random input of $n$ values, $x_1, \dots, x_n$ where each $x_i$
is chosen iid and uniformly in $(0,1)$.
By Yao's minimax principle, the lower bound on the expected \frag of the minimum when running~$\A$ also holds for any randomized algorithm.

We prove our lower bound in a model that we call ``comparisons with additional information (CAI)'':
if the algorithm $\A$ compares two elements $x_i$ and $x_j$ and it turns out that $x_i < x_j$, then the value $x_j$ is revealed to the algorithm.
Clearly, the algorithm can only do better with this extra information. 
The heart of the proof is the following lemma which also acts as the 
``base case'' of our proof.
\begin{lemma}\label{lem:minlb_base}
  Let $\Delta$ be an upper bound on $\frem$.
  Consider $T$ values $x_1, \dots, x_T$ chosen iid and uniformly in $(0,b)$.
  Consider a deterministic algorithm $\A$ in CAI model that finds the minimum value $y$ among $x_1, \dots, x_T$. 
  If $T > 1000\Delta$, then with probability at least $\frac{7}{10}$ $\A$ will compare $y$ against an element~$x$ such that $x \ge b/(100\Delta).$
\end{lemma}
\begin{proof}
  By simple scaling, we can assume $b=1$.
  Let $p$ be the probability that $\A$ compares $y$ against a value larger than $1/(100\Delta)$. 
  Let $I_\sm$ be the set of indices $i$ such that $x_i < 1/(100\Delta)$. 
  Let $\A'$ be a deterministic algorithm in CAI model such that:
  \begin{itemize}
  \item $\A'$ is given all the indices  in $I_\sm$ (and their corresponding values) except for the index of the minimum.
    We call these the known values. 
  \item $\A'$ minimizes the probability $p'$ of comparing the $y$ against a value larger than $1/(100\Delta)$.
  \item $\A'$ finds the minimum value among the unknown values. 
  \end{itemize}
  Since $p' \le p$, it suffices to bound $p'$ from below.
  We do this in the remainder of the proof. 

  Observe that the expected number of values $x_i$ such that $x_i < 1/(100\Delta)$ is $T/(100\Delta)$.
  Thus, by Markov's inequality, $\p[|I_\sm| \le T/(10\Delta) ] \ge \frac{9}{10}$.
  Let's call the event $|I_\sm| \le T/(10\Delta) $ the good event. 
  For algorithm $\A'$ all values smaller than $1/(100\Delta)$ except for the minimum are known.
  Let~$U$ be the set of indices of the unknown values. Observe that a value~$x_i$ for~$i\in U$ is either the minimum or larger than $1/(100\Delta)$, and that $|U|=T-|I_\sm|+1 > \frac9{10} T$ (using $\Delta\geq 1$) in the good event.
  Because $\A'$ is a deterministic algorithm, the set~$U$ is split into set~$F$ of elements that have their first comparison against a known element, and set $W$ of those that are first compared with another element with index in~$U$.
  Because of the global bound~$\Delta$ on the fragile complexity of the known elements, we know $|F|<\Delta\cdot |I_\sm| \leq \Delta T / (10\Delta) = T/10$. %
  Combining this with the probability of the good event, by union bound, the probability of the minimum being compared with a value greater than~$1/(100\Delta)$ is at least
  $1- (1-\frac{9}{10}) - (1-\frac{8}{9}) \ge 7 /10$. 
\end{proof}

Based on the above lemma, our proof idea is the following. 
Let $G = 100\Delta$. 
We would like to prove that on average  $\A$ cannot avoid comparing the minimum to a lot of elements. 
In particular, we show that, with constant probability, the minimum will be compared against some value in the range $[G^{-i}, G^{-i+1}]$ for every integer $i$, $1 \le i \le \frac{\log_G n}{2}$.
Our lower bound then follows by an easy application of the linearity of expectations. 
Proving this, however, is a little bit tricky. 
However, observe that Lemma~\ref{lem:minlb_base} already proves this for $i=1$.
Next, we use the following lemma to apply Lemma~\ref{lem:minlb_base} over all values of $i$, $1 \le i \le \frac{\log_G n}{2}$.
\begin{lemma}\label{lem:minlb_view}
  For a value $b$ with $0 < b < 1$, define $p_k = {n \choose k} b^i (1-b)^{n-k}$, for $0 \le k \le n$. 
  Choosing $x_1, \dots, x_n$ iid and uniformly in $(0,1)$ is equivalent to the following:
  with probability $p_k$, uniformly sample a set $I$ of $k$ distinct indices in $\left\{ 1, \dots, n \right\}$ among all the subsets of size $k$.
  For each $i \in I$, pick $x_i$ iid and uniformly in $(0,b)$.
  For each $i \not \in I$, pick $x_i$ iid and uniformly in $(b,1)$.
\end{lemma}
\begin{proof}
    It is easy to see that choosing $x_1, \dots, x_n$ iid uniformly in $(0,1)$ is equivalent to choosing a point $X$ uniformly at random inside an $n$ dimensional unit cube $(0,1)^n$. Therefore, we will prove the equivalence between (i) the distribution defined in the lemma, and (ii) choosing such point $X$.

    Let $Q$ be the $n$-dimensional unit cube. 
    Subdivide $Q$ into $2^n$ rectangular region defined by the Cartesian product of intervals
    $(0,b)$ and $(b,1)$, i.e., $\left\{ (0,b), (b,1) \right\}^n$ 
    (or alternatively, bisect $Q$ with $n$ hyperplanes, with the $i$-th hyperplane perpendicular to the $i$-th axis and intersecting it at coordinate equal to $b$). 

    Consider the set $R_k$ of rectangles  in  $\left\{ (0,b), (b,1) \right\}^n$ with exactly
    $k$ sides of length $b$ and $n-k$ sides of length $1-b$.
    Observe that for every choice of $k$ (distinct) indices $i_1, \dots, i_k$ out of $\left\{ 1, \dots, n \right\}$,
    there exists exactly one rectangle $r$ in $R_k$ such that $r$ has side length $b$ at dimensions $i_1, \dots, i_k$, 
    and all the other sides of $r$ has length $n-k$. 
    As a result, we know that  the  number of rectangles in $R_k$ is ${n\choose k}$ and the volume of 
    each rectangle in $R_k$ is $b^k(1-b)^k$.
    Thus, if we choose a point $X$ randomly inside $Q$, with probability $p_k$ it will fall inside
    a rectangle $r$ in $R_k$; furthermore, conditioned on this event, 
    the dimensions $i_1, \dots, i_k$ where $r$ has side length $b$ is a 
    uniform subset of $k$ distinct indices from $\left\{ 1, \dots, n \right\}$. 
\end{proof}

Remember that our goal was to prove that with constant probability, the minimum will be compared against some value in the range $[G^{-i}, G^{-i+1}]$ for every integer $i$, $1 \le i \le \frac{\log_G n}{2}$.
We can pick $b=G^{-i+1}$ and apply Lemma~\ref{lem:minlb_view}.
We then observe that it is very likely that the set of indices $I$ that we are sampling in Lemma~\ref{lem:minlb_view} will contain many indices. 
For every element $x_i$, $i\in I$, we are sampling $x_i$ independently and uniformly in $(0,b)$ which opens the door for us
to apply Lemma~\ref{lem:minlb_base}. 
Then we argue that Lemma~\ref{lem:minlb_base} would imply that with constant probability
the minimum will be compared against a value in the range $(b/G,b) = (G^{-i},G^{-i+1})$.
The lower bound claim of Theorem~\ref{thm:exp-min-lower} then follows by invoking the linearity of expectations. 

We are ready to prove that the minimum element will have $\Omega(\log_\Delta n)$ comparisons on average.
\begin{proof}[Proof of Theorem~\ref{thm:exp-min-lower}]
    First, observe that we can assume $n \ge (100,000 \Delta)^2$ as otherwise we are aiming for a trivial bound of  $\Omega(1)$. 
    We create an input set of $n$ values $x_1, \dots, x_n$ where each $x_i$ is chosen iid and uniformly
    in $(0,1)$. 
    Let $G=100\Delta$. 
    Consider an integer $i$ such that $1 \le i < \frac{\log_G n}{2}$.
    We are going to prove that with constant probability, the minimum will be compared against a value in the range $(G^{-i},G^{-i+1})$, which,
    by linearity of expectation, shows the stated $\Omega(\log_\Delta n)$ lower bound for the \frag of the minimum.

    Consider a fixed value of $i$. 
    Let $S$ be the set of indices with values that are smaller than $G^{-i+1}$. 
    Let $p$ be the probability that $\A$ compares the minimum against an~$x_j$ with $j\in S$ such that $x_j \ge G^{-i}$. 
    To prove the theorem, it suffices to prove that $p$ is lower bounded by a constant.
    Now consider an algorithm $\A'$ that finds the minimum but for whom all the values other than those in $S$ have
    been revealed and furthermore, assume $\A'$ minimizes the probability of comparing the minimum against an element $x \ge G^{-i}$ (in other words, we pick the algorithm which minimizes this probability, among all the algorithms). 
    Clearly, $p' \le p$. 
    In the rest of the proof we will give a lower bound for $p'$. 

    Observe that $|S|$ is a random variable with binomial distribution. 
    Hence $\E{|S|} = nG^{-i+1} > \sqrt{n}$ where the latter follows from $i < \frac{\log_G n}{2}$.
    By the properties of the binomial distribution we have that 
    $\p\left[|S| <  \frac{\E{|S|}}{100}\right]< \frac {1}{10}.$ 
    Thus, with probability at least $\frac {9}{10}$, we will have the ``good'' event that
    $|S|\ge \frac{\E{|S|}}{100} \ge \frac{\sqrt{n}}{100}$.

    In case of the good event, Lemma~\ref{lem:minlb_view} implying that conditioned on $S$ being the set of values smaller than $G^{-i+1}$, each value~$x_j$ with $j\in S$ is distributed independently and uniformly in the range $(0,G^{-i+1})$. 
    As a result, we can now invoke Lemma~\ref{lem:minlb_base} on the set $S$ with $T= |S|$.
    Since $n \ge (100,000 \Delta)^2$ we have $T = |S| \ge \frac{\sqrt{n}}{100} \ge \frac{100,000\Delta}{100}$.
    By Lemma~\ref{lem:minlb_base}, with probability at least $\frac{7}{10}$, the minimum will be compared against a value that is larger than $G^{-i}$.
    Thus, by law of total probability, it follows that in case of a good event, with probability $\frac{7}{10}$ the minimum will be compared to a value in the range $(G^{-i},G^{-i+1})$.
    However, as the good event happens with probability $\frac{9}{10}$, it follows that with probability at least $1-(1-\frac{7}{10})-(1-\frac{9}{10})= \frac{6}{10}$, the minimum will be compared against a value in the range $(G^{-i},G^{-i+1})$.
\end{proof}

\subsubsection{Lower bound for the \frag of the minimum whp.}%

With Theorem~\ref{thm:loglogWHP} in Subsection~\ref{sec:randAlgMin}, we show in particular that {\sc SampleMinimum} guarantees that the \frag of the minimum is at most $\Oh(\log\log n)$ with probability at least $1-1/n^c$ for any $c>1$.
(By setting $t = 2 \log\log n$).

Here we show that this is optimal up to constant factors in the fragile complexity.
\begin{theorem}\label{thm:minWHPlb}
  For any constant $\varepsilon > 0$, there exists a value of $n_0$ such that  the following holds for any 
  randomized algorithm $\A$ and for any $n > n_0$: there exists an input of size $n$ such that
  with probability at least $n^{-\varepsilon}$, $\A$ performs  $\ge  \frac12 \log\log n$
  comparisons with the minimum. 
\end{theorem}
\begin{proof}
  We use (again) Yao's principle and consider a fixed deterministic algorithm~$\A$ working on the uniform input distribution, i.e., all input permutations have probability $1/n!$.
  Let $f= \frac12 \log\log n$ be the upper bound on the \frag of the minimum. 
  Let $k=2^{f}=\sqrt{\log n}$ and let $S$ be the set of the $k$ smallest input values. 
  Let $\pi$ be a uniform permutation (the input) and  
  $\pi(S)$ be the permutation of the elements of $S$ in $\pi$.
  Observe that $\pi(S)$ is a uniform permutation of the elements of $S$. 
  We reveal the elements not in $S$ to $\A$. 
  So, $\A$ only needs to find the minimum in $\pi(S)$. 
  By Theorem~\ref{thm:min-det-fragile-theta} there is at least one
  ``bad'' permutation of~$S$ which forces
  algorithm $\A$ to do $\log k = f$ comparisons on the smallest element.
  Observe $\log k! < \log k^k = k \log k = \sqrt{\log n}\frac12\log\log n$. Observe that there exists
  a value of $n_0$ such that for $n>n_0$ the right hand side is upper bounded by $\varepsilon \log n$, 
  so $k! \leq n^\varepsilon$, for $n > n_0$.
  Hence, the probability of a ``bad'' permutation is at least $1/k! > n^{-\varepsilon}$.
  \qedhere
\end{proof}
 
\section{Selection and median}\label{sec:main-selection}
The {\em \selection{t} problem} asks to find the $t$-th smallest element among $n$ elements of the input.
The simplest solution to the \selection{t} problem is to sort the input.
Therefore, it can be solved in $\Oh(\log n)$ \frag and $\Oh(n\log n)$ work by using the AKS sorting network~\cite{aks-halvers}.
For comparator networks, both of these bounds are optimal: the former is shown by Theorem~\ref{thm:min-det-fragile-theta} (and in fact it applies also to any algorithm) and the latter is shown in Section~\ref{sec:deterministicSelectionNetwork}.

In contrast, in this section we show that comparison-based algorithms can do better:
we can solve \select deterministically in $\Theta(n)$ work and $\Theta(\log n)$ \frag, thus, showing a separation between the two models.
However, to do that, we resort to constructions that are based on expander graphs.
Avoiding usage of the expander graphs or finding a simpler optimal deterministic solution is an interesting open problem (see Section~\ref{sec:conc}).
Moreover, in Subsection~\ref{sec:median} we show that we can do even better by using randomization.

\subsection{Deterministic selection}
\label{sec:deterministicSelectionAlg}

\begin{theorem}\label{thm:detMed}
There is a deterministic algorithm for \select which performs $\Oh(n)$ \work and has $\Oh(\log n)$ \frag.
\end{theorem}
\begin{proof}
Below, we give an algorithm for the median problem. By simple padding of
the input, median solves the \selection{t} problem for arbitrary
$t \neq \frac{n}{2}$.

A central building block from the AKS sorting network is an
\emph{$\varepsilon$-halver}. An $\varepsilon$-halver approximately
performs a partitioning of an array of size~$n$ into the smallest half
and the largest half of the elements. More precisely, for any
$m \le n/2$, at most~$\varepsilon n$ of the $m$ smallest elements will
end up in the right half of the array, and at most~$\varepsilon n$ of
the $m$ largest elements will end up in the left half of the
array. Using expander graphs, a comparator network implementing an
$\varepsilon$-halver in constant depth can be
built~\cite{aks-halvers,focs92*686}. We use the corresponding comparison-based algorithm of constant fragile complexity.

We make the convention that when using $\varepsilon$-halvers, the larger elements
placed at the right half and the smaller elements are placed at the left half. 
We first use an $\varepsilon$-halver on the input array~$S$ of length
$n$, dividing it into two subarrays of length~$n/2$. 
Let's call the right half $S_1$. 
As an $\varepsilon$-halver does an ``approximate'' partitioning, $S_1$ will contain
$l$ of the smallest $n/2$ elements by mistake, however, it will contain 
$n/2-l$ of the $n/2$ largest with $l \le \varepsilon n/2$. 
From this point forward, we apply $\varepsilon$-halvers but alternate between picking the right and then left half.
In particular, we apply an $\varepsilon$-halver to $S_i$ (starting from $i=1$), and set $S_{i+1}$ to be the 
left half (resp. right half) of the resulting partition if $i$ is odd (resp. if $i$ is even). 
See Figure~\ref{fig:recursiveHalving}.

We stop the process after $k=2\lfloor\frac{\log\log n}{2} \rfloor $ steps; we choose an even $k$ to simplify
the upcoming discussions. 
This results in a set  $S_k$ of size $\Theta(n/\log n)$. 
We sort $S_k$ using an AKS-based sorting algorithm, which takes $\Oh(n)$ work and
has \frag $\Oh(\log n)$, and we then extract the middle $|S_k|/2$ of
these sorted elements as the set $R_P$ (``right pivots''). 

We claim the rank of every element in $R_P$ is between $(1+\alpha)\frac{n}{2}$
and $(2-\alpha)\frac{n}{2}$ for some absolute constant $0 < \alpha < 1$.
To prove this claim, we use the properties of $\varepsilon$-halvers. 
Consider $S_i$: we partition $S_i$ into two sets, and select $S_{i+1}$ to be either
the left or the right half, depending on the parity of $i$.
Assume $S_{i+1}$ is selected to be the right half (similarly, left half). 
We mark an element of $S_{i+1}$ as a \emph{left mistake} (similarly, a \emph{right mistake}), 
if it is among the $|S_i|/2$ smaller (similarly, larger) elements of $S_i$. 
We say an element of $S_{i+1}$ is \emph{good}.

Now assume $\varepsilon < 1/64$.
We can now use simple induction to show the following:
If $S_i$ contains $t_\ell$ left mistakes and $t_r$ 
right mistakes, then the left mistakes are the $t_\ell$ smallest element of $S_i$
and the right mistakes are the $t_r$ largest elements in $S_i$ and furthermore, 
$S_i$ contains at most $2\varepsilon|S_i|$ left mistakes and 
$2\varepsilon |S_i|$ right mistakes. 

These claims are obviously true for $S_1$ and   $S_2$ thus assume the hold for $S_i$;
we would like to show that they also hold for $S_{i+2}$. 
W.l.o.g, assume $S_{i+1}$ is selected to be the right half after partitioning $S_i$ using
an $\varepsilon$-halver. 
Consider the sorted order of $S_i$ and in particular, the  set $L$ containing $|S_i|/2$ 
largest elements in $S_i$. 
As $\varepsilon < 1/64$, it follows that $2\varepsilon |S_i|< |S_i|/2$ and as a result, 
$L$ contains no left mistakes by our induction hypothesis since all the left mistakes are among the 
$|S_i|/2$ smallest elements of $|S_i|$. 
However, $L$  contains all the up to $2\varepsilon |S_i|$ right mistakes. 
By properties of an $\varepsilon$-halver, $S_{i+1}$ has at most $\varepsilon|S_i|$ elements
that do not belong in $L$. 
Thus, $S_{i+1}$ contains at most $\varepsilon|S_i|$ left mistakes and 
at most $2\varepsilon|S_i|$ right mistakes.
Crucially, notice that $S_{i+2}$ is obtained by using an $\varepsilon$-halver on $L$ and
selecting \emph{the left} half of the resulting partition. 
A similar argument now shows that $S_{i+2}$ has $\varepsilon|S_i| =2 \varepsilon|S_{i+1}|$ left mistakes
but $\varepsilon|S_{i+1}|$ right mistakes.
This concludes the inductive proof of our claims. 
Observe that, as a corollary, at least $1-4/\varepsilon > 1/2$ fraction of the elements
in $S_i$ are good. 

\begin{figure}[t]
  \centering
  \includegraphics[width=0.9\textwidth]{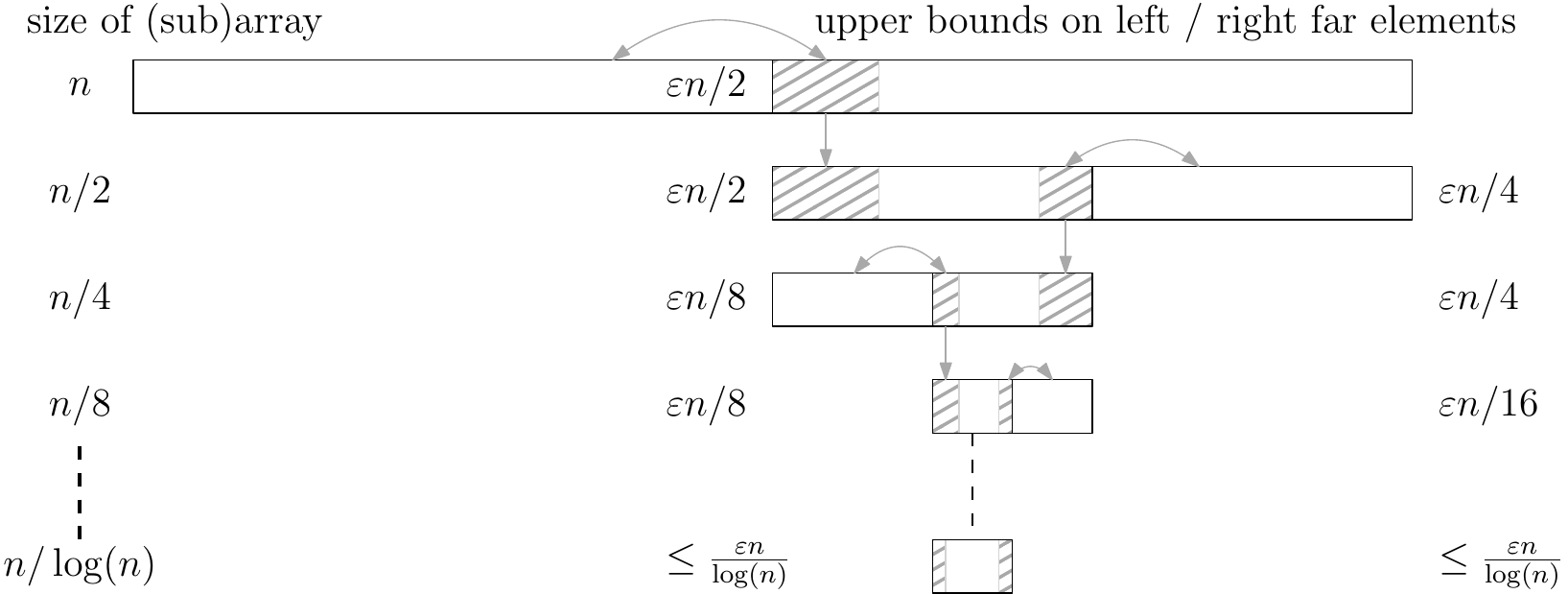}
  \caption{Illustration of the alternating division process using $\varepsilon$-halvers.}
  \label{fig:recursiveHalving}
\end{figure}

After $i$ steps, we have an array part~$S_i$ of length~$n_i=n/2^i$ with
at most $2\varepsilon n_i$ left mistakes and at most $2\varepsilon n_i$ right
mistakes. 
For a moment assume there are no mistakes in any of the partitioning steps done by 
$\varepsilon$-halvers. 
An easy inductive proof implies that in this case, the rank of the elements in $S_i$, for even $i$, are between 
$A_i = (1+\frac{4^{i-1}-1}{3 4^{i-1}})\frac{n}{2}$ and 
$B_i = (1+\frac{4^{i-1}+3}{3 4^{i-1}})\frac{n}{2}$. 
The claim is clearly true for $i=2$ and it can be verified for $i+2$:
there are $\frac{n}{2} \frac{1}{4^{i-1}}$ elements between the aforementioned ranks
and thus, the ranks of the elements in $S_{i+2}$ will between
$A_i+\frac{n}{2} \frac{1}{4^{i}}$  and 
$A_i+ 2\frac{n}{2} \frac{1}{4^{i}}$ (partition the range between $A_i$ and $B_i$ into four equal chunks
and pick the second chunk).
Now observe that $A_i +\frac{n}{2} \frac{1}{4^{i}} = (1+\frac{4^{i-1}-1}{3 4^{i-1}} + \frac{1}{4^i})\frac{n}{2}
= (1+\frac{4^i-1}{3 4^i})\frac{n}{2}$ and $A_i +2\frac{n}{2} \frac{1}{4^{i}} =
(1+\frac{4^{i-1}-1}{3 4^{i-1}} + 2\frac{1}{4^i})\frac{n}{2}
= (1+\frac{4^i+3}{3 4^i})\frac{n}{2}$.

However, $\varepsilon$-halvers will likely make plenty of mistakes. 
Nonetheless, observe that every mistake made by an $\varepsilon$-halver can only change the rank of a good element $x$
in $S_k$ by one: each mistake either involves placing an element that is actually smaller than $x$ to the right of $x$
or an element that is larger than $x$ to the left of $x$. 
Furthermore, observe that the total number of elements marked as a mistake is bounded by a geometric series:
\[ 
    \sum_{i=1}^\infty 4\varepsilon \frac{n}{2^i} \le 4 \varepsilon n < \frac{n}{16}.
\]

\noindent This in turn implies that the rank of the good elements in $S_k$ is between
\begin{align}
    (1+\frac{4^{k-1}-1}{3 4^{k-1}})\frac{n}{2}-\frac{n}{16} = (1+ \frac{1}{3}-\frac{1}{34^{k-1}} - \frac{1}{8})\frac{n}{2}  > (1+\frac{1}{8})\frac{n}{2}\label{eq:left}
\end{align}
and 
\begin{align}
    (1+\frac{4^{k-1}+3}{3 4^{k-1}})\frac{n}{2}+\frac{n}{16} = (1+ \frac{1}{3}+\frac{1}{4^{k-1}} + \frac{1}{8})\frac{n}{2} < (1+\frac{7}{8})\frac{n}{2}\label{eq:right}.
\end{align}

Remember that $S_k$ contains $\Theta(n/\log n)$ good elements.
We might not know exactly which elements of $S_k$ are good but we know that
the middle $|S_k|/2$ elements are certainly good.
As discussed, we select these elements as our pivot set $R_P$.
We now compare all elements~$y$ of~$S$ with some element~$x$ from $R_P$.
We evenly distribute the comparisons such that every element of 
$R_P$ gets compared against at most $\lceil n/|R_P|\rceil = O(\log n)$
elements.
If $x \ge y$, we mark $y$ by~\markR. An element
marked by~\markR has rank in~$S$ of at least $(1+1/8)n/2$ by Eq.~\ref{eq:left}, and thus it is guaranteed 
not to be the median. 
By Eq.~\ref{eq:right}, we mark at least $n/16$ elements by~\markR.

We then perform a symmetric process in the first half of~$S$, leading to
a set $\overline{S}_k$ of size $\Theta(n/\log n)$ which has a subset $L_P$ (for
``left pivots'') of $\Theta(n/\log n)$ elements whose rank in~$S$ are 
between $\frac{1}{8}\cdot \frac{n}{2}$ and $\frac{7}{8}\cdot \frac{n}{2}$. 
We compare all elements~$y$ of~$S$ with
some element~$x'$ from $L_P$ and as before distribute the comparisons evenly among the elements of  
$L_P$.
If $y < x'$, we mark $y$ by~\markL. 
And as before, an element marked by~\markL is smaller than the median and 
at least $n/16$ elements are marked by~\markL.

We can discard an element of mark~\markL together with an element of mark~\markR
as doing so will not change the median among the remaining elements, as we are discarding
an element larger than the median and an element smaller than the median. 
As a result, we can discard $n/8$ elements.

At first glance, it might feel like we are done. However, we cannot safely recurse
on the remaining elements as the elements in $S_k$ and $\overline{S}_k$ have already incurred too many
comparisons. 
This results in a slight complication but it can be averted as follows.

At the beginning we have an input $S$ of $n$ elements. 
Define $S\tin{0} = S$ and $n\tin{0}=n$ as the top level of our recursion.
At depth $i$ of the recursion, we have a set $S\tin{i}$ containing $n\tin{i}$ elements. 
We then select sets $S\tin{i}_k$ and $\overline{S}\tin{i}_k$ (note that $k$ is also a function of
$i$ but to reduce the clutter in the notation, we just use $k$) and discard
$n\tin{i}/8$ elements of $S\tin{i}$ using the above mentioned procedure. 
The set $S\tin{i+1}$ is defined as the set containing the remaining elements, \emph{excluding}
the elements of the sets $S\tin{i}_k$ and $\overline{S}\tin{i}_k$; the elements of these sets are not
pruned but they are put aside and we will handle them later.

We stop the recursion as soon as we reach a recursion depth $j$ with  $n\tin{j} \le n/\log n$. 
At this point, we simply union all the elements in the sets $S\tin{i}_k$ and $\overline{S}\tin{i}_k$,
$1 \le i < j$ and $S\tin{j}$ into a final set $\S$, sort $\S$ and then report the median of $\S$. 
As discussed, pruning elements does not change the median and thus the reported median is correct.
Note that since we always prune a fraction of the elements at each recursive step, we have 
$j = O(\log\log n)$. 

We now analyze the \frag. Invoking an $\varepsilon$-halver incurs $\Oh(1)$
comparisons on each element, so a full division process incurs
$\Oh(\log\log n)$ comparisons on each element of $S\tin{i}$.
Over $j = O(\log \log n)$ recursions, this adds up to $O( (\log\log n)^2) = O(\log n)$. 
Next, the elements in $S\tin{i}_k$ and $\overline{S}\tin{i}_k$ are sorted but they are not sent to the
next recursion step so this is simply a one-time cost. 
These elements suffer an additional $\Theta(\log n)$ comparisons during the pruning phase
(to be precise, a subset of them that are selected as ``left pivot'' or ``right pivot'' elements)
but this cost is also only suffered once.
Finally, the remaining elements participate in a final sorting round. 
Thus, each element participates in $O(\log n)$ comparisons in the worst-case.

Thus, it remains to analyze the \work. 
Invoking an $\varepsilon$-halver incurs a linear number of comparisons. 
However, as we prune at least a fraction of the elements at each step,
we have $n\tin{i} \le n (\frac{7}{8})^i$. 
Thus, the total amount of work done by $\varepsilon$-halvers is linear. 
The same holds for the sorting of $S\tin{i}_k$ and $\overline{S}\tin{i}_k$
as their sizes is bounded by $O(n\tin{i}/\log n\tin{i})$ which forms a geometrically decreasing series. 
Finally, observe that $|\S| = O(n/\log n)$ since
$S\tin{j} = O(n/\log n)$ and 
$\sum_{i=0}^j |S\tin{i}_k| + |\overline{S}\tin{i}_k| = O(n/\log n)$.
his implies, we can sort $\S$ in $\Oh(n)$ work as well. 
\end{proof}

\begin{corollary}
There is a deterministic algorithm for partition which performs $\Oh(n)$
\work and has $\Oh(\log n)$ \frag.
\end{corollary}
\begin{proof}
At the end of the \select algorithm, the set of elements smaller (larger) than the median is the union of the respective filtered sets (sets $\markL$ and $\markR$ in the proof in the full version of the paper~\cite{DBLP:journals/corr/abs-1901-02857}) and the first (last) half of the sorted set in the base case of the recursion.
Again, simple padding generalizes this to \partition{t} for arbitrary $t \neq \frac{n}{2}$.
\end{proof}
\subsection{Deterministic selection via comparator networks}
\label{sec:deterministicSelectionNetwork}

In this section, we discuss the \selection{t} problem in the setting of comparator networks. We present an upper and a matching lower bound on the size of a comparator network solving the \mbox{\selection{t}} problem. In the next section, we consider the problem in the setting of comparison-based algorithms and give an algorithm for selection with the same fragile complexity as in the case of comparator networks, but with total work that is asymptotically smaller. Combined, this shows a separation in power between the two models.

To begin, observe that the \selection{t} problem can be solved by sorting the input. Therefore, the \selection{t} problem can be solved using a comparator network of size $\Oh (n \log n)$ and depth $\Oh(\log n)$ by using the AKS sorting network~\cite{aks}. Consequently, the \selection{t} problem can be solved with $\Oh(\log n)$ \frag and $\Oh(n \log n)$ \work.

Next, we show that the size of the \selection{t} network using the AKS network is asymptotically tight. Before we present the lower bound theorem, we need to introduce some notation and prove two auxiliary lemmas. 
 
Given a set $X$, we say two elements $x_i, x_j \in X$ are {\em rank-neighboring} if $|\rank{x_i} - \rank{x_j}| = 1$. We say a permutation $\hat{\pi}(X)$ is {\em rank-neighboring} if an ordered sequence $X$ and $\hat{\pi}(X)$ differ in only two elements and these two elements are rank-neighboring.
Observe that any permutation $\pi(X)$ is a composition of some number of rank-neighboring permutations.

We define the {\em signature} of an ordered sequence $X$ with respect to an integer $a$ to be a function $\sig: \Z \times \Z^{n-1} \rightarrow \{0, 1\}^{n-1}$, such that $\sig(a, (x_1, \dots, x_{n-1} )) = (y_1, \dots, y_{n-1})$ and for all $1\le i\le n-1$: 
\begin{eqnarray*}
y_i = \left\{\begin{array}{ll}
    0 &\mbox{ if } x_i \le a \\
    1 &\mbox{ if } x_i  > a 
  \end{array}
\right.
\end{eqnarray*}

\begin{lemma}\label{lemma:rank-neighbor}
For any totally ordered set $X$ of size $n$, any \selection{t} network $\networknt$, and for any rank-neighboring permutation $\hat{\pi}$: $\sig(\networknt(X))=\sig(\networknt(\hat{\pi}(X)))$. 
\end{lemma}
\begin{proof}
Consider the two inputs $x_i$ and $x_j$ in which $X$ and $\hat{\pi}(X)$ differ. During the computation the input values traverse the network until they reach the outputs. During the computation of the network $\networknt(X)$ element $x_i$ (resp., $x_j$) starts at the $i$-th (resp., $j$-th) input and reaches the $i'$-th (resp., $j'$-th) output. During the computation of the network $\networknt(\hat{\pi}(X))$, element $x_i$ (resp., $x_j$) starts at the $j$-th (resp., $i$-th) input. Let us determine which outputs they reach.

Consider the two paths $P_i$ and $P_j$ that $x_i$ and $x_j$, respectively, traverse during the computation of $\networknt(X)$. 
Since $\hat{\pi}$ is a rank-neighboring permutation, the outputs of every comparator $C$ are the same for $\networknt(X)$ and $\networknt(\hat{\pi}(X))$ except for the (set of) comparator(s) $C^*$, whose two inputs are $x_i$ and $x_j$. 
Thus, throughout the computation of $\networknt(\hat{\pi}(X))$,  $x_i$ and $x_j$ only traverse the edges of the paths $P_i \cup P_j$, i.e., the outputs of $\networknt(\hat{\pi}(X))$ are the same as the outputs of $\networknt(X)$ everywhere except for, possibly, at the $i'$-th and $j'$-th output. 

If $\rank{x_i} < t$ and $\rank{x_j} < t$, or $\rank{x_i} > t$ and $\rank{x_j} > t$, the signatures of these two outputs are the same. Consequently, $\sig(\networknt(X))=\sig(\networknt(\hat{\pi}(X)))$.  
Otherwise, without loss of generality, let $\rank{x_i} = t$ (the case of $\rank{x_j} = t$ is symmetric). 
Then $i' = 0$, and $\networknt(X)^{0} = \networknt(\hat{\pi}(X))^{0} = x_i$. Then, $\networknt(X)^{j'} = \networknt(\hat{\pi}(X))^{j'} = x_j$, and it follows that $\sig(\networknt(X))=\sig(\networknt(\hat{\pi}(X)))$.
\end{proof}

\begin{lemma}\label{lemma:rank-network}
An \selection{t}  network can be turned into an \partition{t} network.
\end{lemma}
\begin{proof}
Since every permutation $\pi(X)$ can be obtained from $X$ by a sequence of rank-neighboring permutations, it follows from Lemma~\ref{lemma:rank-neighbor} that $\sig(\networknt(\pi(X))) = \sig(\networknt(X))$, i.e., for every permutation of the inputs in the \selection{t} network, the same subset of outputs carry the values that are at most $t$. 
Thus, the \partition{t} network can be obtained from the \selection{t} network by reordering (re-wiring) the outputs such that the ones with signature $0$ are the first $t$ outputs. 
\end{proof}

\begin{theorem}
An \selection{t} network for any $t \le n/2$ has size $\Omega((n-t) \log (t+1))$.
\end{theorem}
\begin{proof}
Alekseev~\cite{alekseev:selection-69} showed the $\Omega((n-t)\log(t+1))$ lower bound for the size of an \partition{t} network, and, by Lemma~\ref{lemma:rank-network}, every \selection{t} network also solves the \partition{t} problem.
\end{proof}

\begin{corollary}
A comparator network that finds the median of $n$ elements has size $\Omega(n \log n)$.
\end{corollary}
\subsection{Randomized selection}
\label{sec:median}

\fbox{\scalebox{0.98}{\begin{minipage}{\textwidth}
\begin{algorithmic}[1]
  \Procedure{\randmedname}{$X=\{x_1, \ldots, x_n\}, k(\cdot), d(\cdot)$}
    \Comment{Sampling phase}  
    \State Randomly sample $k$ elements from $X$, move them into an array $S$, sort $S$ with AKS
    \State Choose appropriate~$b$ to distribute $S$ into buckets $L_b, \ldots L_1, C, R_1, \ldots,  R_b$ such that:
    \State \hspace{1em} $n_0 = 2 \sqrt{k \log n}$, $n_1 = 3 \sqrt{k \log n}$, $n_i = d \cdot n_{i-1}$
    \State \hspace{1em} $C = S[ k/2  {-} n_0     : k/2 {+} n_0 ]$  median candidates
    \State \hspace{1em} $L_i = S[k/2 {-} n_i     : k/2 - n_{i-1}]$ buckets of elements presumed smaller than median
    \State \hspace{1em} $R_i = S[k/2 {+} n_{i-1} : k/2 {+} n_{i}]$ buckets of elements presumed larger than median
  
    \Comment{Probing phase}
    \For{$x_i \in X \setminus S$ \text{in random order}}
      \For{$j \in [b{-}1, \ldots, 1]$ \text{in order}}
        \State $x_A \gets $ arbitrary element in $L_j$ with fewest compares
        \State $c \gets 1$ \text{if $x_A$ is marked else} $2$
		\Comment{Pivots added in probing (= weak) are marked}
        \If{$x_i < x_A$}
          \State add $x_i$ as new pivot to $L_{j+c}$ if $j < b-c$ and mark it, \\\ \hspace{7em} otherwise discard $x_i$ by inserting it into $L_b$
          \State continue with next element $x_{i+1}$
        \EndIf
        
        \State $x_B \gets $ arbitrary element in $R_j$ with fewest compares
        \State $c \gets 1$ \text{if $x_B$ is marked else} $2$
        \If{$x_i > x_B$}
          \State add $x_i$ as new pivot to $R_{j+c}$ if $j < b-c$ and mark it, \\\ \hspace{7em} otherwise discard $x_i$ by inserting it into $R_b$
          \State continue with next element $x_{i+1}$
        \EndIf
      \EndFor
    
      \Comment{By now it is established that $S[k/2 - n_1] \le x_i \le S[k/2 + n_1]$}
      \State add $x_i$ as a median candidate to $C$
    \EndFor
  
    \If{$\max(\sum_i |L_i|, \sum_i |R_i|) > n/2$}
      \Comment{Partitioning too imbalanced $\Rightarrow$ median not in $C$}
      \label{alg:rand-median-fallback}
      \State \Return \Call{DetMedian}{X}
    \EndIf
  
    \If{$|C| < \log^4 N$ where $N$ is the size of the initial input}
      \State sort $C$ with AKS and return median
    \EndIf
  
    \State $j = \sum_i \left(|L_i| - |R_i|\right)$
  
    \If{$j < 0$}
      \State add $|j|$ arbitrary elements from $\bigcup_i R_i$ to $C$
    \Else
      \State add $j$ arbitrary elements from $\bigcup_i L_i$ to $C$
    \EndIf
  
    \State \Return \Call{\randmedname}{$C$, $k(\cdot)$, $d(\cdot)$}
  \EndProcedure
\end{algorithmic}
\end{minipage}}}
 \vspace{1em}

We now present the details of an expected work-optimal selection algorithm with a trade-off between the expected fragile complexity $\fmed(n)$ of the selected element and the maximum expected fragile complexity $\frem(n)$ of the remaining elements.
In particular, we obtain the following combinations:
\begin{theorem}\label{thm:RandSelectionCombined}
Randomized selection is possible in expected linear work, while achieving expected \frag of the median $\E{\fmed(n)} = \Oh(\log\log n)$ and of the remaining elements $\E{\frem(n)} = \Oh(\sqrt{n})$, or $\E{\fmed(n)} = \Oh\left(\frac{\log n}{\log\log n}\right)$ and $\E{\frem(n)} = \Oh(\log^2 n)$.
\end{theorem}
Just like with the deterministic approach in Section~\ref{sec:deterministicSelectionAlg} we restrict ourselves to the special case of median finding.  The general $(n,t)$-selection problem can be solved by
initially adding an appropriate number of distinct dummy elements. Note that comparisons of real elements with dummy elements do not contribute to the fragile complexity since we know that all dummy elements are either larger or smaller than all real elements depending on the value of $t$. Similarly, the comparison between dummy elements
comes for free.

We present \randmed{} an expected work-optimal median selection algorithm with a trade-off between the expected fragile complexity $\fmed(n)$ of the median element and the maximum expected fragile complexity $\frem(n)$ of the remaining elements. 
By adjusting a parameter affecting the trade-off, we can vary $\left\langle\E{\fmed(n)}, \max\E{\frem(n)}\right\rangle$ in the range between $\left\langle\Oh(\log\log n), \Oh(\sqrt n)\right\rangle$ and $\left\langle \Oh(\log n/\log\log n), \Oh(\log^2 n) \right\rangle$ (see Theorems~\ref{thm:randmedLogLog} and \ref{thm:randmedSubLog}). 

\randmed{} (see pseudo code) takes a totally ordered set $X$ as input.
It draws a set $S$ of $k(n)$ random samples, sorts them and subsequently uses the items in $S$ as pivots to identify a set $C \subset X$ of values around the median, such that an equal number of items smaller and greater than the median are excluded from $C$. %
Finally, it recurses on $C$ to select the median.

The recursion reaches the base case when the input is of size $\Oh(\polylog n)$, %
at which point it can be sorted to trivially expose the median.
\randmed{} employs two sets $L_1$ and $R_1$ of $n_1 = \Oh(\sqrt{k \log n})$ pivots almost surely below and above the median respectively.
All candidates for $C$ are compared to one item in $L_1$ and $R_1$ each filtering elements that are either too small or too large.
The sizes of $L_1, R_1$ and $C$ are balanced to achieve fast pruning and low failure probability\footnote{%
  \randmed{} is guaranteed to select the correct median.
  Failure results in an asymptotically insignificant increase of expected fragile complexity (see line~\ref{alg:rand-median-fallback} in the pseudocode).
} (see Lemmas~\ref{lem:rand-median-not-in-c} and~\ref{lem:rand-median-rec-depth}).

To reduce the fragile complexity of elements in $L_1$ and $R_1$, most elements are prefiltered using a cascade of weaker classifiers $L_i$ and $R_i$ for $2 \le i \le b$ geometrically growing in size by a factor of~$d(n)$ when moving away from the center.
Filtered elements that are classified into a bucket $L_i$ or $B_i$ with $i < b$ are used as new pivots and effectively limit the expected load per pivot.
As the median is likely to travel through this cascade, the number of filter layers is a compromise between the fragile complexity \fmed{} of the median and \frem{} of the remaining elements.

We define $k(n)$ and $d(n)$ as functions since they depend on the problem size which changes for recursive calls.
If $n$ is unambiguous from context, we denote them as $k$ and $d$ respectively and assume $k(n) = \Omega(n^\varepsilon)$ for some $\varepsilon > 0$.

\begin{lemma}\label{lem:rand-median-not-in-c}
  Consider any recursion step and let $X$ be the set of $n$ elements passed as the subproblem. After all elements in $X$ are processed, the center partition $C$ contains the median $x_m \in X$ whp.
\end{lemma}
\begin{proof}%
  The algorithm can fail to move the median into bucket $C$ only if the sample $S \subset X$ is highly skewed.
  More formally, we use a safety margin $n_0$ around the median of $S$ and observe that if there exists $x_l, x_r \in S_C := S[k/2 {-} n_0 : k/2 + n_0 ]$ with $x_l < x_m < x_r$, the median $x_m$ is moved into~$C$.

  This fails in case too many small or large elements are sampled.
  In the following, we bound the probability of the former from above;
  the symmetric case of too many large elements follows analogously.  
  Consider $k$ Bernoulli random variables $X_i$ indicating that the $i$-th sample $s_i < x_m$ lies below the median and apply Chernoff's inequality $\pr{\sum_i X_i > (1{+}\delta) \mu} \le \exp(- \mu \delta^2/3 )$ where $\mu = \E{\sum_i X_i}$ and $\delta < 1$:
  \begin{equation*}
    \pr{\lnot \exists\ x_r {\in} S_C\colon x_m {<} x_r}
    = \pr{\sum_{i=1}^k X_i > k/2 {+} n_0}
    \le \exp\left(- \frac{2n_0^2}{3k}\right) 
    \stackrel{n_0 = \sqrt {2k \log n}}{=} n^{-4/3} \,. \qedhere
  \end{equation*}
\end{proof}

\begin{lemma}\label{lem:rand-median-rec-depth}
  Each recursion step of \randmed{} reduces the problem size from $n$ to $\Oh(\sqrt{n \log n})$ whp.
\end{lemma}
\begin{proof}%
  The algorithm recurses on the center bucket $C$ which contains the initial sample of size $2n_0 = \Oh(\sqrt{n \log n})$ and all elements that are not filtered out by $L_1$ and $R_1$.
  We pessimistically assume that each element added to $C$ compared to the weakest classifiers in these filters (i.e., the largest element in $R_1$ and the smallest in $L_1$).
  
  We hence bound the rank in $X$ of the $R_1$'s largest pivot; due to symmetry $L_1$ follows analogously.
  Using a setup similar to the proof of Lemma~\ref{lem:rand-median-not-in-c}, we define Bernoulli random variables indicating that the $i$-th sample $s_i$ is larger than the $\ell$-th largest element in $X$ where $\ell = n/2 + 3\sqrt{n\log(n)}$.
  Applying Chernoff's inequality yields the claim.
\end{proof}

\begin{lemma}\label{lem:rand-median-frag-med}
  The expected \frag of the median is
  \[ \E{\fmed(n)} = \E{\fmed\left(\Oh(\sqrt {n \log n})\right)} + \Oh\big(
      \underbrace{\frac k n \log k}_\text{Sampled} +
      \underbrace{(1-\frac k n) \log_d k}_\text{Not sampled} +
      \underbrace{\phantom{\frac k n}1\phantom{\frac k n}}_\text{Misclassified} 
  \big)\,. \]
\end{lemma}
\begin{proof}%
  Due to Lemma~\ref{lem:rand-median-rec-depth}, a recursion step reduces the problem size from $n$ to $\Oh(\sqrt{n\log n})$ whp represented in the recursive summand.
  The remaining terms apply depending on whether the median is sampled or not:
  with $\pr{x_m \not\in S} = 1 - k/n$ the median is not sampled and moved towards the center triggering a constant number of comparisons in each of the $\Oh(\log_d n)$ buckets.
  Otherwise if the median $x_m$ is sampled, it incurs $\Oh(\log k)$ comparisons while $S$ is sorted.
  By Lemma~\ref{lem:rand-median-not-in-c}, the median is then assigned to $C$ whp and protected from further comparisons.
  According to Lemma~\ref{lem:rand-median-not-in-c}, the complementary event of $x_m$ being misclassified has a vanishing contribution due to its small probability.
\end{proof}

\begin{lemma}\label{lem:rand-median-frag-rem}%
  The expected \frag of non-median elements is 
  \begin{align*}
    \E{\frem(n)} &= \E{\frem\left(\Oh(\sqrt {n \log n})\right)} \\
    &+ \Oh\big(
      \underbrace{\log k}_\text{Sampled} +
      \underbrace{\log_d n}_\text{Not sampled} +
      \max(
          \underbrace{\phantom{\frac k n}d^2\phantom{\frac k n}}_{\scriptstyle\begin{array}{c}\text{Pivot in $R_i$} \\ i {>} 2\end{array}},
          \underbrace{\frac{nd}{k}}_{\scriptstyle\begin{array}{c}\text{Pivot in $R_j$} \\ j {\le} 2\end{array}}
      )
    \big)\,,
  \end{align*} 
  where $d(n) = \Omega(\log^\varepsilon n)$ for some $\varepsilon>0$ and $k(n) = \Oh(n / \log n)$.
\end{lemma}

\begin{proof}
  As the recursion is implemented analogously to Lemma~\ref{lem:rand-median-frag-med}, we discuss only the contribution of a single recursion step.
  Let $x \in X {\setminus} \{x_m\}$ be an arbitrary non-median element.
  
  The element $x$ is either sampled and participates in $\Oh(\log k)$ comparisons.
  Otherwise it traverse the filter cascade and moves to $C$, $L_i$ or $R_i$ requiring $\Oh(\log_d n)$ comparisons.
  
  If it becomes a median candidate (i.e. $x \in C$), $x$ has a fragile complexity as discussed in Lemma~\ref{lem:rand-median-frag-med} which is asymptotically negligible here.
  Thus we only consider the case that $x$ is assigned to $L_i$ or $R_i$ and we assume $x \in R_i$ without loss of generality due to symmetry.
  If it becomes a member of the outer-most bucket $R_b$, it is effectively discarded.
  Otherwise, it can function as a new pivot element replenishing the bucket's comparison budget.
  As \randmed{} always uses a bucket's least-frequently compared element as pivot, it suffices to bound the expected number of comparisons until a new pivot arrives.
    
  Observe that \randmed{} needs to find an element $y \in (R_{i-2} \cup R_{i-1})$ with $y < x$ in order to establish that $x \in R_i$.
  This is due to the fact that pivots can be placed near the unfavorable border of a bucket rendering them weak classifiers.
  We here pessimistically assume that $y \in R_{i-2}$ for simplicity's sake.
  By construction the initial bucket sizes $n_i$ grow geometrically by a factor of $d$ as $i$ increases.
  Therefore, any item compared to bucket $R_i$ continues to the next bucket with probability at most $1 / d$.
  Consequently, bucket $R_i$ with $i>2$ sustains expected $\Oh(d^2)$ comparisons until a new pivot arrives.
  
  This is not true for the two inner most buckets $R_j$ with $j \in \{1,2\}$ as they are not replenished.
  Bucket $R_j$ ultimately receives $\Oh(n \cdot \frac{n_j}{k})$ items whp, however it is expected to process $d$~times as many comparisons due to the possibly weak classifiers in the previous bucket $R_{j+1}$.
  Since bucket $R_j$ contains $n_j$ pivots in , each of it participates in $\Oh(nd / k)$ comparisons.
\end{proof}

\begin{theorem}\label{thm:randmedLogLog}
  \randmed{} achieves $\E{\fmed(n)} = \Oh(\log\log n)$ and $\E{\frem(n)} = \Oh(\sqrt{n})$.
\end{theorem}
\begin{proof}
  Choose $k(n) = n^\varepsilon$, $d(n) = n^\delta$ with $\varepsilon=2/3$, $\delta = 1/12$. Then Lemmas~\ref{lem:rand-median-frag-med} and~\ref{lem:rand-median-frag-rem} yield:
  \begin{align*}
    \E{\fmed(n)} &= \E{\fmed(\Oh(\sqrt{n \log n}))} + \Oh(n^{\varepsilon - 1} \varepsilon \log n + \frac \varepsilon \delta) 
    = \Oh(\log\log n)\,, \\
    \E{\frem(n)} &= \E{\frem(\Oh(\sqrt{n \log n}))} + \Oh(
      (\varepsilon {+} \frac 1 \delta) \log n + 
      n^{\max(2\delta, 1 - \varepsilon + 2\delta)})
    = \Oh(\sqrt n) \,. \qedhere
  \end{align*}
\end{proof}

\begin{theorem}\label{thm:randmedSubLog}
  \randmed{} achieves $\E{\fmed(n)} = \Oh\left(\frac{\log n}{\log\log n}\right)$, $\E{\frem(n)} = \Oh(\log^2 n)$.
\end{theorem}
\begin{proof}
  Choose $k(n) = \frac{n}{\log n}$, $d(n) = \log n$. Then Lemmas~\ref{lem:rand-median-frag-med} and~\ref{lem:rand-median-frag-rem} yield:
  \begin{align*}
    \E{\fmed(n)} &= \E{\fmed(\Oh(\sqrt{n \log n}))} + \Oh\left(
      \frac{\log n}{\log \log n}
    \right) = \Oh\left(\frac{\log n}{\log\log n}\right)\,, \\
    \E{\frem(n)} &= \E{\frem(\Oh(\sqrt{n \log n}))} + \Oh(
      \log^2 n
    ) = \Oh(\log^2 n) \,. \qedhere
  \end{align*}
\end{proof}
  
\begin{theorem}\label{lem:randmed-work}
  For $k = \Oh(n/\log n)$ and $d = \Omega(\log n)$, \randmed{} performs a total of $\Oh(n)$ comparisons in expectation, implying $\w{n} = \Oh(n)$ expected work.
\end{theorem}
\begin{proof}
  We consider the first recursion step and analyze the total number of comparisons.
  \randmed{} sorts $k$ elements using AKS resulting in $\Oh(k \log k) = \Oh(n)$ comparisons. 
  It then moves $\Oh(n)$ items through the filtering cascade consisting of buckets of geometrically decreasing size resulting of $\Oh(1)$ expected comparisons per item.
  Each bucket stores its pivots in a minimum priority queue with the number of comparisons endured by each pivot as keys.
  Even without exploitation of integer keys, retrieving and inserting keys is possible with $\Oh(\log n)$ comparisons each.
  Hence, we select a pivot and keep it for $d = \Omega(\log n)$ steps, resulting in amortized $\Oh(1)$ work per comparison.
  This does not affect $\fmed$ and $\frem$ asymptotically.
  Using Lemma~\ref{lem:rand-median-rec-depth}, the total number of comparisons hence $g(n) = g(\sqrt{n\log n}) + \Oh(n) = \Oh(n)$.
\end{proof}

\section{Sorting}\label{sec:main-sorting}
Recall from Section~\ref{sec:intro} that the few existing sorting
networks with depth $\Oh(\log n)$ are all based on expanders, while a
number of $\Oh(\log^2 n)$ depth networks have been developed based on
binary merging. Here, we study the power of the mergesort paradigm
with respect to fragile complexity. We first prove that any sorting
algorithm based on binary merging must have a worst-case fragile
complexity of $\Omega(\log^2 n)$.  This provides an explanation why
all existing sorting networks based on merging have a depth no better
than this. We also prove that the standard mergesort algorithm on
random input has fragile complexity $\Oh(\log n)$ with high
probability, thereby showing a separation between the deterministic
and the randomized situation for binary mergesorts. Finally, we
demonstrate that the standard mergesort algorithm has a worst-case
fragile complexity of $\Theta(n)$, but that this can be improved to
$\Oh(\log^2 n)$ by changing the merging algorithm to use exponential
search.

\begin{lemma}\label{lem:merge}
  Merging of two sorted sequences $A$ and $B$ has fragile complexity at least
  $\lfloor \log_2 |A| \rfloor + 1$.
\end{lemma}
\begin{proof}
  A standard adversary argument: The adversary designates one element $x$
  in $B$ to be the scapegoat and resolves in advance answers to
  comparisons between $A$ and $B_{1} = \{y \in B \mid y < x\}$ by
  $B_{1} < A$ and answers to comparisons between $A$ and
  $B_{2} = \{y \in B \mid x < y\}$ by $A < B_{2}$. There are still $|A|+1$
  total orders on $A \cup B$ compatible with these choices, one for each
  position of $x$ in the sorted order of~$A$. Only comparisons between $x$
  and members of $A$ can make some of these total orders incompatible with
  answers given by the adversary. Since the adversary can always choose to
  answer such comparisons in a way which at most halves the number of
  compatible orders, at least $\lfloor \log_2 |A| \rfloor + 1$ comparisons
  involving $x$ have to take place before a single total order is known.
\end{proof}

By standard \textsc{MergeSort}, we mean the algorithm which divides the
$n$~input elements into two sets of sizes $\lceil n/2 \rceil$ and
$\lfloor n/2 \rfloor$, recursively sorts these, and then merges the
resulting two sorted sequences into one. The merge algorithm is not
restricted, unless we specify it explicitly (which we only do for the
upper bound, not the lower bound).

\begin{lemma}\label{lem:balancedmergesort}
  Standard \textsc{MergeSort} has fragile complexity $\Omega(\log^2 n)$.
\end{lemma}
\begin{proof}
  In \textsc{MergeSort}, when merging two sorted sequences $A$ and $B$, no
  comparisons between elements of $A$ and $B$ have taken place before the
  merge. Also, the sorted order of $A \cup B$ has to be decided by the
  algorithm after the merge. We can therefore run the adversary argument
  from the proof of Lemma~\ref{lem:merge} in all nodes of the mergetree of
  \textsc{MergeSort}. If the adversary reuses scapegoat elements in a bottom-up
  fashion---that is, as scapegoat for a merge of $A$ and $B$ chooses one
  of the two scapegoats from the two merges producing $A$ and $B$---then
  the scapegoat at the root of the mergetree has participated in
  \[\Omega(\sum_{i=0}^{\log n} \log 2^i ) = \Omega(\sum_{i=0}^{\log n} i
  ) = \Omega(\log^2 n)\]
  comparisons, by Lemma~\ref{lem:merge} and the fact that a node at
  height~$i$ in the mergetree of standard \textsc{MergeSort} operates on sequences
  of length $\Theta(2^i)$.
\end{proof}
We now show that making unbalanced merges cannot improve the fragile
complexity of binary \textsc{MergeSort}.
\begin{theorem}\label{thm:mergesort-lower}
  Any binary mergesort has fragile complexity $\Omega(\log^2 n)$.
\end{theorem}
\begin{proof}
  The adversary is the same as in the proof of
  Lemma~\ref{lem:balancedmergesort}, except that as scapegoat element for
  a merge of $A$ and $B$ it always chooses the scapegoat from the
  \emph{larger} of $A$ and $B$. We claim that for this adversary, there is
  a constant $c > 0$ such that for any node~$v$ in the mergetree, its
  scapegoat element has participated in at least $c \log^2 n$ comparisons
  in the subtree of~$v$, where $n$ is the number of elements merged
  by~$v$. This implies the theorem.

  We prove the claim by induction on~$n$. The base case is $n = \Oh(1)$,
  where the claim is true for small enough~$c$, as the scapegoat by
  Lemma~\ref{lem:merge} will have participated in at least one
  comparison. For the induction step, assume~$v$ merges two sequences of
  sizes $n_1$ and $n_2$, with $n_1 \ge n_2$. By the base case, we can
  assume $n_1 \ge 3$. Using Lemma~\ref{lem:merge}, we  would like to prove for the induction step
  \begin{equation}
    c\log^2 n_1 + \lfloor \log n_2 \rfloor + 1  \ge c\log^2
    (n_1+n_2).\label{eq:inductionstep}
  \end{equation}

  This will follow if we can prove that
  \begin{equation}
    \log^2 n_1 + \frac{\log n_2}{c} \ge \log^2 (n_1+n_2)\,.\label{eq:inductionstepv2}
  \end{equation}
  The function~$f(x) = \log^2 x$ has first derivative $2 (\log x)/x $ and
  second derivative~$2(1-\log x)/x^2$, which is negative for
  $x > e = 2.71\dots$.  Hence, $f(x)$ is concave for $x > e$, which means
  that first order Taylor expansion (alias the tangent) lies above~$f$,
  i.e., $f(x_0) + f'(x_0)(x-x_0) \ge f(x)$ for $x_0, x > e$. Using
  $x_0 = n_1$ and $x = n_1 + n_2$ and substituting the first order Taylor
  expansion into the right side of (\ref{eq:inductionstepv2}), we see that
  (\ref{eq:inductionstepv2}) will follow if we can prove
  \begin{equation*}
    \frac{\log n_2}{c} \ge
    2\frac{\log n_1}{n_1}n_2\,,
  \end{equation*}
  which is equivalent to
  \begin{equation}
    \frac{\log n_2}{n_2} \ge
    2c\frac{\log n_1}{n_1}\,.\label{eq:inductionstepv3}
  \end{equation}
  Since $n_1 \ge n_2$ and $(\log x)/ x$ is decreasing for $x \ge e$, we see
  that (\ref{eq:inductionstepv3}) is true for $n_2 \ge 3$ and~$c$ small
  enough. Since $\log(3)/3 = 0.366\dots$ and $\log{2}/2 = 0.346\dots$, it is
  also true for $n_2 = 2$ and~$c$ small enough. For the final case of
  $n_2 = 1$, the original inequality~(\ref{eq:inductionstep}) reduces to
  \begin{equation}
    \log^2 n_1 + \frac{1}{c} \ge \log^2(n_1+1)\,.\label{eq:inductionstepv4}
  \end{equation}
  Here we can again use concavity and first order Taylor approximation
  with $x_0 = n_1$ and $x = n_1 + 1$ to argue that
  (\ref{eq:inductionstepv4}) follows from
  \begin{equation*}
    \frac{1}{c} \ge 2\frac{\log n_1}{n_1}\,.
  \end{equation*}
  which is true for $c$ small enough, as $n_1 \ge 3$ and $(\log x)/ x$ is
  decreasing for $x \ge e$.
\end{proof}

\subsection{Upper Bound for MergeSort with Linear Merging}
By \emph{linear merging}, we mean the classic sequential merge algorithm that takes two input sequence and iteratively moves the the minimum of both to the output.

\begin{observation}\label{obs:merging-cost}
  Consider two sorted sequences $A=(a_1, \ldots, a_n)$ and $B=(b_1, \ldots, b_n)$.
  In linear merging, the fragile complexity of element $a_i$ is at most $\ell+1$ where $\ell$ is the largest number of elements from $B$ that are placed directly in front of $a_i$ (i.e.  $b_j < \ldots < b_{j+\ell-1} < a_i$).
\end{observation}

\begin{theorem}
  \label{thm:mergesort-linear}
  Standard \textsc{MergeSort} with linear merging has a worst-case fragile complexity of $\Theta(n)$.
\end{theorem}
\begin{proof}
  \textsl{Lower bound $f(n) = \Omega(n)$}: linear merging requires $\Oh(k)$ comparisons to output a sequence of length $k$.
  In standard \textsc{MergeSort}, each element takes part in $\Oh(\log n)$ merges of geometrically decreasing sizes $n/2^i$ (from root), resulting in $\Oh(n)$ comparisons.
  
  \textsl{Upper bound $f(n) = \Oh(n)$}: consider the input sequence $(n, 1, 2, \ldots, n{-}1)$ where $n = 2^k$.
  Then every node on the the left-most path of the mergetree contains element $n$.
  In each merging step, we receive $A=(1, \ldots, \ell-1, n)$ from the left child, $B=(\ell, \ldots, 2\ell{-}1)$ from the right, and produce 
  \[(
  \underbrace{1, \ldots, \ell{-}1}_{\text{from } A},\
  \underbrace{\ell, \ldots, 2\ell{-}1}_{\text{from } B},\ 
  \underbrace{\phantom{\ell,}n\phantom{\ell,}}_{\text{from } A}
  )\,.
  \]

  Hence, the whole sequence $B$ is placed directly in front of element $n$, resulting in $\Theta(\ell)$ comparisons with this element according to Observation~\ref{obs:merging-cost}.
  Then, the sum of the geometrically increasing sequence length yields the claim.
\end{proof}

\begin{lemma}\label{lem:geom-subset}
  Let $X = \{x_1, \ldots, x_{2k}\}$ be a finite set of distinct elements, and consider a random bipartition $X_L, X_R \subset X$ with $|X_L| = |X_R| = k$ and $X_L \cap X_R = \emptyset$, such that $\pr{x_i \in X_L} = 1/2$.
  Consider an arbitrary ordered set $Y = \{y_1, \ldots, y_m\} \subset X$ with $m \le k$.
  Then $\pr{Y \subseteq X_L \lor Y \subseteq X_R} < 2^{1 - m}$.
\end{lemma}
\begin{proof}
  \begin{equation*}
    \pr{Y \subseteq X_L \lor Y \subseteq X_R}
    = 2\prod_{i=1}^{m} \pr{\,y_i \in X_L \,\middle|\, y_1, \ldots y_{i-1} \in X_L\,}
    = 2\frac{(2k)^{-m} k!}{(k-m)!} \le 2\cdot 2^{-m}. \qedhere
  \end{equation*}
\end{proof}

\begin{theorem}
  \label{thm:mergesort-whp}
  Standard \textsc{MergeSort} with linear merging on a randomized input permutation has a fragile complexity of $\Oh(\log n)$ with high probability.
\end{theorem}
\begin{proof}
  Let $Y=(y_1, \ldots, y_n)$ be the input-sequence, $\pi^{-1}$ be the permutation that sorts $Y$ and $X = (x_1, \ldots, x_n)$ with $x_i = y_{\pi^{-1}(i)}$ be the sorted sequence.
  Wlog we assume that all elements are unique\footnote{%
    If this is not the case, use input sequence $Y' = ((y_1, 1), \ldots, (y_n, n))$ and lexicographical compares.%
  }, that any input permutation $\pi$ is equally likely\footnote{If not shuffle it before sorting in linear time and no fragile comparisons.}, and that $n$ is a power of two.

  \textsl{Merging in one layer}.
  Consider any merging-step in the mergetree.
  Since both input sequences are sorted, the only information still observable from the initial permutation is the bi-partitioning of elements into the two subproblems.
  Given $\pi$, we can uniquely retrace the mergetree (and vice-versa):
  we identify each node in the recursion tree with the set of elements it considers.
  Then, any node with elements $X_P = \{y_\ell, \ldots, y_{\ell+2k-1}\}$ has children
  \begin{align*}
    X_L &= \left\{x_{\pi(i)}\ \middle|\ \ell   \le \pi(i) \le \ell+k-1 \, \right\} = \{y_\ell, \ldots, y_{\ell+k-1}\}\,, \\
    X_R &= \left\{x_{\pi(i)}\ \middle|\ \ell+k  \le  \pi(i) \le \ell+2k-1 \, \right\} = \{y_{\ell+k}, \ldots, y_{\ell+2k-1} \}\,.
  \end{align*}
  Hence, locally our input permutation corresponds to an stochastic experiment in which we randomly draw exactly half of the parent's elements for the left child, while the remainder goes to right.

  This is exactly the situation in Lemma~\ref{lem:geom-subset}.
  Let $N_i$ be a random variable denoting the number of comparisons of element $y_i$ in the merging step.
  Then, from Observation~\ref{obs:merging-cost} and Lemma~\ref{lem:geom-subset} it follows that $\pr{N_i = m{+}1} \le 2^{-m}$. 
  Therefore $N_i$ is stochastically dominated by $N_i \preceq 1{+}Y_i$ where $Y_i$ is a geometric random variable with success probability $p = 1/2$.

  \textsl{Merging in all layers}.
  Let $N_{j,i}$ be the number of times element $y_i$ is compared in the $j$-th recursion layer and define $Y_{j,i}$ analogously.
  Due to the recursive partitioning argument, $N_{j,i}$ and $Y_{j,i}$ are iid in~$j$.
  Let $N^T_i$ be the total number of comparisons of element~$i$, i.e. $N^T_i \preceq \log_2 n + \sum_{j=1}^{\log_2 n} Y_{j,i} $.
  Then a tail bound on the sum of geometric variables (Theorem 2.1 in~\cite{JANSON20181}) yields:
  \begin{equation*}
    \pr{
    \sum_{j=1}^{\log_2 n} Y_{j,i} \, \ge\,  \lambda \E{ \sum_{j=1}^{\log_2 n} Y_{j,i} } =2 \lambda \log_2 n}
      \stackrel{\text{\cite{JANSON20181}}}{\le} 
      \exp\left(-\frac{1}{2}  \frac{2\ln n}{\ln 2} [\lambda {-} 1 {-} \log \lambda] \right)
      = n^{-2},
  \end{equation*}
    where we set $\lambda \approx 3.69$ in the last step solving $\lambda  {-} \log \lambda = 2\log 2$.
    Thus, we bound the probability $\pr{ N^T_i \ge (1 {+} 2\lambda) \log_2 n } \le n^{-2}$.

    \textsl{Fragile complexity}.
    It remains to show that with high probability no element exceeds the claimed fragile complexity.
    We use a union bound on $N^T_i$ for all $i$:
    \begin{equation*}
      \pr{\max_i\{N^T_i\} = \omega(\log n)}
      \le n \pr{N^T_i = \omega(\log n)} \le 1/n\,. \qedhere
    \end{equation*}
\end{proof}

\subsection{Upper Bound for MergeSort with Exponential Merging}
\begin{figure}
  \centering
  \includegraphics[width=0.6\textwidth]{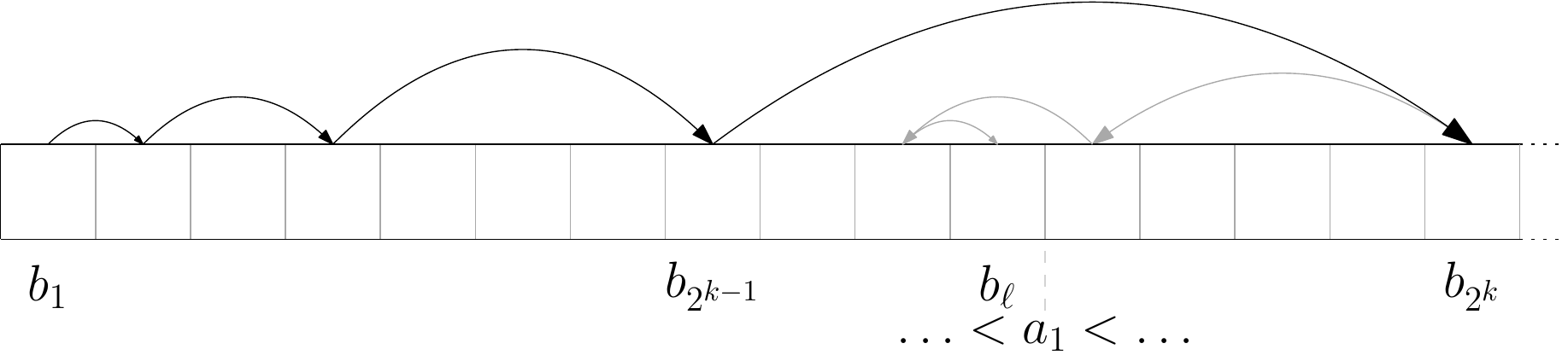}
  \caption{The Exponential search performs $k$ doubling steps and overshoots the target $b_\ell$ with $b_\ell < a_1 < b_{\ell+1}$.
    A binary search between $b_{2^{k-1}}$ and $b_{2^k}$ ultimately identifies $b_\ell$ in $\Oh(k)$ steps.}
  \label{fig:exponentialSearch}
\end{figure}

  We define \emph{exponential merging} of sequences $A=(a_1, \ldots, a_n)$ and $B=(b_1, \ldots, b_m)$ as follows:
  if either $A$ or $B$ are empty, output the other one and stop.
  Otherwise, assume without loss of generality that $m$ is a power of two and that there exists an $b_i \in Y$ with $a_1 < b_i$, if not append sufficiently many virtual elements $b_\top$ to $B$ with $a_1 < b_\top$.
  Use an exponential search on $B$ starting in $b_1$ to find all elements $b_1 < \ldots < b_\ell < a_1$ smaller than $a_1$.
  As illustrated in Fig.~\ref{fig:exponentialSearch}, the exponential search consists of a \emph{doubling phase} which finds the smallest $k$ with $a_1 < b_{2^k}$.
  Since the doubling phase may overshoot $b_\ell$, a \emph{binary search} between $b_{2^{k-1}}$ and $y_{b^k}$ follows.
  Output $b_1, \ldots, b_\ell, a_1$ and recurse on $A'=[b_{\ell+1}, \ldots, b_m]$ and $B'=[a_2, \ldots, a_n]$ which swaps the roles of $A$ and $B$.
\begin{theorem}\label{thm:exp_merging}
  Exponential merging of two sequences $A=(a_1, \ldots, a_n)$ and $B=(b_1, \ldots, b_n)$ has a worst-case fragile complexity of $\Oh(\log{n})$.
\end{theorem}
\begin{proof}
  Without loss of generality let $n$ be a power of two and consider a single exponential search finding the smallest $k$ with $a_1 < b_{2^k}$.
  The element $a_1$ is compared to all $\{ b_{2^i} \,|\, 1 \le i \le k \}$ during the doubling phase.
  We use an accounting argument to bound the \frag.
  Element $a_1$ takes part in every comparison and is charged with $k = \Oh(\log n)$.
  It is then charged $\Oh(\log n)$ comparisons during the binary search between $b_{2^{k-1}}$ and $b_{2^k}$.
  It is then moved to the output and not considered again.
  
  The search also potentially interacts with the $2^k$ elements $b_1, \ldots, b_{2^k}$ by either comparing them during the doubling phase, during the binary search or by skipping over them.
  We pessimistically charge each of these elements with one comparison.
  It then remains to show that no element takes part in more than $\Oh(\log n)$ exponential searches.
  
  Observe that all elements $b_1, \ldots, b_{2^{k-1}}$ are moved to the output and do not take part in any more comparisons.
  In the worst-case, the binary search proves that element $b_{2^{k-1}+1}$ and its successors are larger than $a_1$.
  Hence at most half of the elements covered by the exponential search are available for further comparisons.
  To maximize the charge, we recursively setup exponential search whose doubling phases ends in $b_{2^l}$ yielding a recursion depth of $\Oh(\log n)$.\footnote{%
    If the following searches were shorter they would artificially limit the recursion depth.
    If they were longer, too many elements are removed from consideration as only the binary search range can be charged again.%
  }
\end{proof}

\begin{restatable}{corollary}{corMergesortExponential}
  \label{cor:expmergesort}
  Applying Theorem~\ref{thm:exp_merging} to standard \textsc{MergeSort} with exponential merging yields a fragile complexity of $\Oh(\log^2 n)$ in the worst-case.
\end{restatable}

\section{Constructing Binary Heaps}
\label{sec:heapconstruction}

\begin{theorem}\label{obs:heapconstruction}
  The \frag of the standard binary heap construction algorithm of Floyd~\cite{floyd:heap-64} is $\Oh(\log n)$.
\end{theorem}
\begin{proof}
  Consider first an element sifting down along a path in the tree: as the binary tree being heapified has height $\Oh(\log(n))$ and the element moving down is compared to one child per step, the cost to this element before it stops moving is $\Oh(\log(n))$.
  Consider now what may happen to an element $x$ in the tree as another element $y$ is sifting down: $x$ is only hit if $y$ is swapped with the parent of $x$ which implies that $y$ was an ancestor of $x$.
  As the height of the tree is $\Oh(\log(n))$, at most $\Oh(\log(n))$ elements reside on the path above $x$.
  Note that the $x$ may be moved up once as $y$ passes by it; this only lowers the number of elements above $x$.
  In total, any element in the heap is hit at most $\Oh(\log(n))$ times during heapify.
\end{proof}
We note that this \frag is optimal by Theorem~\ref{thm:min-det-fragile-theta}, since \textsc{Heap Construction} is stronger than \textsc{Minimum}.
Brodal and Pinotti~\cite{brodal:heap-98} showed how to construct a binary heap using a comparator network in $\Theta(n\log\log n)$ size and $\Oh(\log n)$ depth.
They also proved a matching lower bound on the size of the comparator network for this problem.
This, together with Observation~\ref{obs:heapconstruction} and the fact that Floyd's algorithm has \work $\Oh(n)$, gives a separation between work of fragility-optimal comparison-based algorithms and size of depth-optimal comparator networks for \textsc{Heap Construction}.
 
\section{Conclusions}
In this paper we introduced the notion of \frag of comparison-based algorithms and we argued that 
the concept is well-motivated because of connections both to real world situations (e.g., sporting events), as well as
other fundamental theoretical concepts (e.g., sorting networks).
We studied the \frag of some of the fundamental problems and revealed interesting behavior such as the
large gap between the performance of deterministic and randomized algorithms for finding the minimum.
We believe there are still plenty of interesting and fundamental problems left open.
Below, we briefly review a few of them.
\begin{itemize}
  \item The area of comparison-based algorithms is much larger than what we have studied. In particular,
    it would be interesting to study ``geometric orthogonal problems'' such as finding the maxima of a set of points, 
    detecting intersections between vertical and horizontal line segments, $kd$-trees, axis-aligned point location and so on.
    All of these problems can be solved using algorithms that simply compare the coordinates of points. 
\item Is it possible to avoid using expander graphs to obtain simple deterministic algorithms to find the median or to sort?
\item Is it possible to obtain a randomized algorithm that finds the median where the median suffers $O(1)$ comparisons
on average? Or alternatively, is it possible to prove a lower bound?
If one cannot show a $\omega(1)$ lower bound for the \frag of the median, can we show it for some other similar problem?
\end{itemize}

\label{sec:conc}

\bibliography{refs}

\end{document}